\documentclass[12pt,a4paper,twoside]{report}
\usepackage[vmargin=20mm,hmargin=27mm]{geometry}
\usepackage{graphicx}
\usepackage{parskip}
\usepackage{setspace}
\usepackage{refcount}
\usepackage{upquote}
\usepackage{amsmath, amssymb, amsthm}
\usepackage{url}
\usepackage[dvipsnames]{xcolor}

\usepackage{graphicx}
\usepackage{tikz}
\usepackage{svg}
\usepackage{mdframed}
\usepackage[colorlinks=true,
linkcolor=webgreen,
filecolor=webbrown,
citecolor=webgreen, pdfborder={0 0 0}]{hyperref}
\usepackage[ruled,vlined]{algorithm2e}
\usepackage[T1]{fontenc}
\usepackage{libertine}
\usepackage[libertine]{newtxmath} % Use Libertine for math
\usepackage[scaled=1]{inconsolata}
\usepackage{comment}

\usepackage{longtable}
\usepackage[tableposition=below]{caption}
\definecolor{webgreen}{rgb}{0,.5,0}
\definecolor{webbrown}{rgb}{.6,0,0}

\newtheorem{theorem}{Theorem}[section]
\newtheorem{lemma}[theorem]{Lemma}
\newtheorem{conjecture}[theorem]{Conjecture}
\newtheorem{corollary}[theorem]{Corollary}
\newtheorem{example}[theorem]{Example}
\newtheorem{definition}{Definition}
\newtheorem{observation}[theorem]{Observation}
\newtheorem{problem}[theorem]{Problem}

\newif\ifsubmission
\title{Approaching the Conway-99 problem using SAT solvers}
\author{Ali Keramatipour}
\date{June 2023}
\newcommand{\candidatenumber}{9634J}
\newcommand{\college}{Churchill College}
\newcommand{\course}{Master of Philosophy in Advanced Computer Science}

\newcommand{\red}[1]{\texttt{\textcolor{red!80!black}{#1}}}
\newcommand{\syntactic}[1]{\texttt{#1}}

\begin{document}

\begin{sffamily} % use a sans-serif font for the pro-forma cover sheet

\begin{titlepage}
\makeatletter

% University logo with shield hanging in left margin
\hspace*{-14mm}\includegraphics[width=65mm]{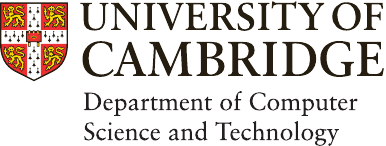}

\ifsubmission

% submission proforma cover page for blind marking
\begin{Large}
\vspace{20mm}
Research project report title page

\vspace{35mm}
Candidate \candidatenumber

\vspace{42mm}
\textsl{``\@title''}

\end{Large}

\else

% regular cover page
\begin{center}
\Huge
\vspace{\fill}

\@title
\vspace{\fill}

\@author
\vspace{10mm}

\Large
\college
\vspace{\fill}

\@date
\vspace{\fill}

\end{center}

\fi

\vspace{\fill}
\begin{center}
Submitted in partial fulfillment of the requirements for the\\
\course
\end{center}

\makeatother
\end{titlepage}

\newpage

Total page count: \pageref{lastpage}

% calculate number of pages from
% \label{firstcontentpage} to \label{lastcontentpage} inclusive
\makeatletter
\@tempcnta=\getpagerefnumber{lastcontentpage}\relax%
\advance\@tempcnta by -\getpagerefnumber{firstcontentpage}%
\advance\@tempcnta by 1%
\xdef\contentpages{\the\@tempcnta}%
\makeatother

Main chapters (excluding front-matter, references and appendix):
\contentpages~pages
(pp~\pageref{firstcontentpage}--\pageref{lastcontentpage})

Main chapters word count: 13675

Methodology used to generate that word count: Overleaf

\end{sffamily}

\vspace{\fill}
% \onehalfspacing
\ifsubmission\else\makeatletter
\textbf{\Huge Declaration}
\vspace{40pt}

I, \@author\ of \college, being a candidate for the \course, hereby
declare that this report and the work described in it are my own work,
unaided except as may be specified below, and that the report does not
contain material that has already been used to any substantial extent
for a comparable purpose.

% Add here things like: Figure X is the work of Y, etc.

\bigskip 
\textbf{Signed:} Ali Keramatipour

\bigskip
\textbf{Date:} 2023/06/18
\vspace{\fill}
\makeatother\fi

\chapter*{Abstract}

The Conway-99 problem questions the existence of a strongly regular graph with 99 vertices and specific parameters. A \textit{strongly} regular graph is a regular graph that exhibits two additional properties: vertices must share a fixed number of neighbours, depending on whether they are adjacent or not, given by two parameters. Despite the search space for this graph being finite, the computational power needed to traverse it is substantial. Therefore, better strategies are required in order to find this graph or prove its non-existence. SAT solvers, designed to solve instances of boolean satisfiability formulas, have been developed and optimised significantly due to the simplicity of SAT problems. Based on Cook-Levin's theorem, computer scientists have been focusing on developing efficient SAT solvers as many problems can be reduced to a SAT problem instance. Hence, we decided to approach the Conway-99 problem using SAT solvers. To do this, we study strongly regular graphs' properties and SAT solvers' capabilities. By encoding the problem of finding strongly regular graphs into SAT instances and running experimental tests, we shall see the incapability of SAT solvers facing this problem in a reasonable time. We will then explore the underlying mathematical reasons for these limitations.

% \ifsubmission\else
% % not included in submission for blind marking:

% \chapter*{Acknowledgements}

% This project would not have been possible without the wonderful
% support of \ldots [optional]

% \fi
\cleardoublepage % preserve page numbers after missing acknowledgements

\tableofcontents
%\listoffigures
%\listoftables

\chapter{Introduction}
\label{firstcontentpage}

The Conway-99 problem \cite{conway5}, an open math problem in the field of graph theory, asks about the existence of a strongly regular graph with 99 vertices and specific parameters. Despite being named after John Horton Conway, J. Seidel was the one who initially showed interest in discovering this graph, as mentioned in a dedication to him \cite{dedicatedToSeidel}. In 1969, Norman L. Biggs suggested the possible existence of this graph with the said parameters \cite{biggs1971finite}.

\section{Strongly regular graphs}
A regular graph is a graph where each vertex has the same degree (number of neighbours).
A strongly regular graph is a regular graph with the property that the number of common
neighbours of two vertices can be determined specifically by their adjacency. The notion of strongly regular graphs was first introduced by Raj Chandra Bose \cite{bose1963strongly}, a mathematician known for his work in the theory of error-correcting codes.

Mathematicians have conducted extensive research on the existence of various strongly regular graphs. Their symmetrical properties have proved to be valuable and applicable in numerous fields, including statistics, Euclidean geometry, group theory, coding theory, cryptography, and more. For example, these symmetrical structures meet coding theory when building efficient and easy-to-decode codes. Crnković et al. \cite{crnkovic2012ternary} have developed various ternary codes using the structure of strongly regular graphs. Alternatively, in statistics, one can use these ideally structured graphs for experiments \cite{Bro22}. One distinctive example is \cite{janmark2014global}, where Janmark et al. made use of strongly regular graphs' local symmetrical structure in quantum search.

\begin{figure}[htbp]
  \centering
  \begin{minipage}[t]{0.4\textwidth}
    \centering
  \includegraphics[width=\linewidth]{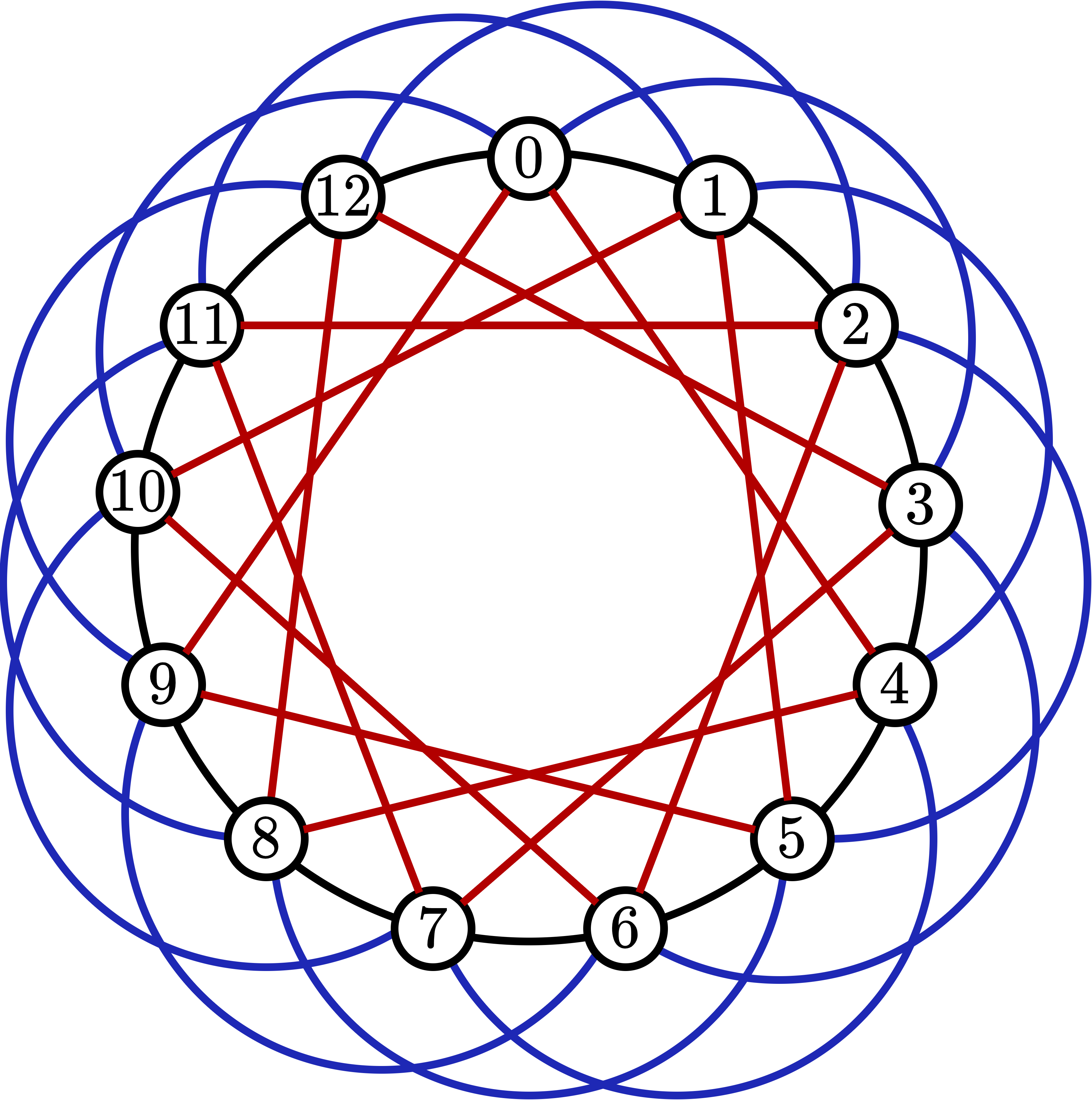}
  \caption{A strongly regular graph, Paley(13).}
  \label{paley13}
  \end{minipage}\hfill
  \begin{minipage}[t]{0.4\textwidth}
  \centering
  \includegraphics[width=\linewidth]{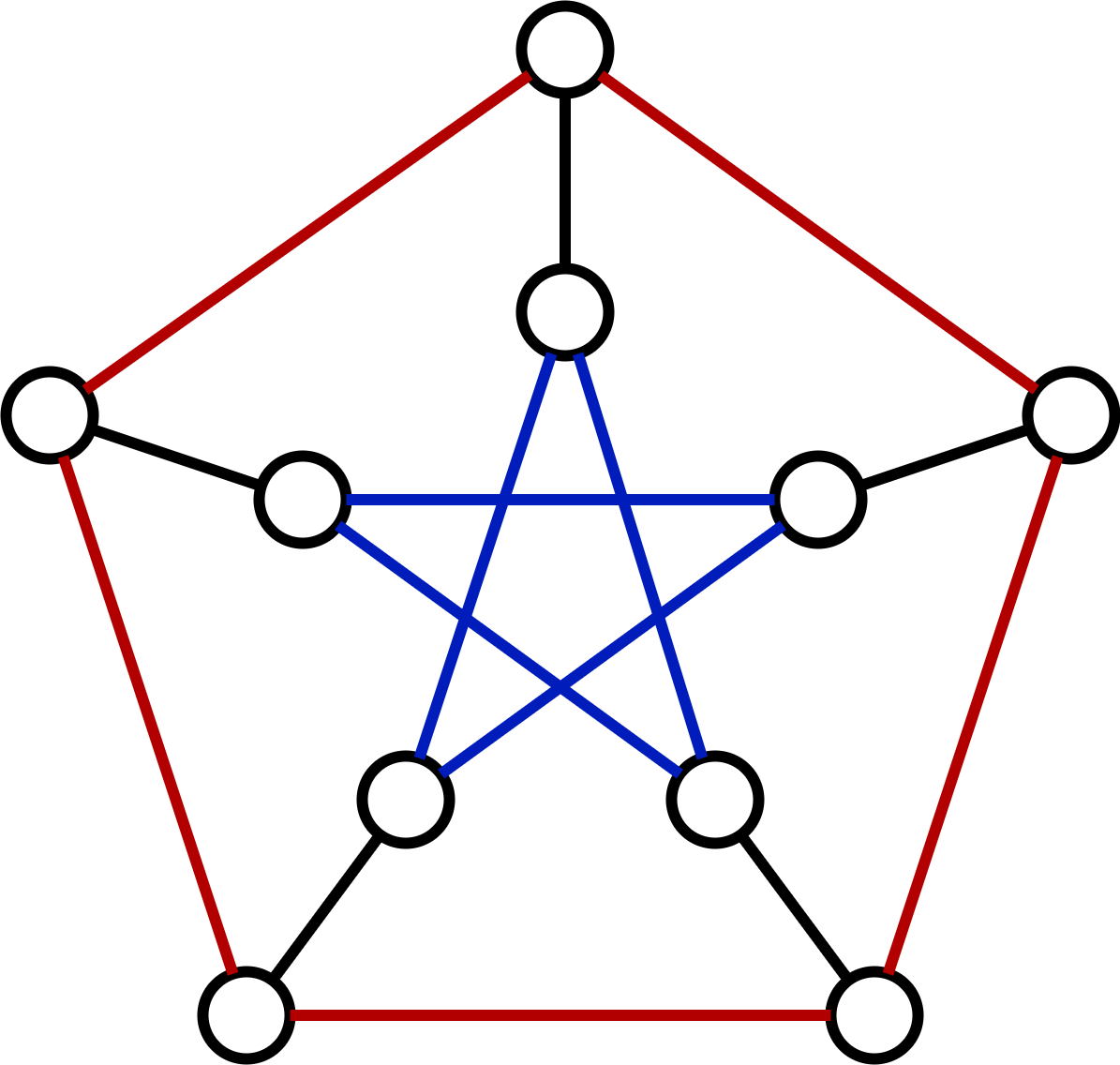}
  \caption{Petersen graph}
  \label{fig:p1}
  \end{minipage}
\end{figure}

These highly structured graphs are so interesting because they also show random behaviours \cite{cameron2004strongly}. In fact, this randomness makes it hard to prove the existence or non-existence of such graphs. Based on both random and structured properties of strongly regular graphs, many applications can be discovered.

One famous and small strongly regular graph, with ten vertices, is the Petersen graph shown in Figure \ref{fig:p1}. In \cite{knuth1997art}, Donald Knuth describes the Petersen graph as \textit{``a remarkable configuration that serves as a counterexample to many optimistic predictions about what might be true for graphs in general"}. This graph is a perfect example of strongly regular graphs' strange and valuable nature. An entire book, authored by Holton and Sheehan \cite{holton1993petersen}, has been dedicated to it because of its usefulness.

It can be proven that at most five strongly regular graphs exist, such that every edge belongs to a unique triangle and every non-edge belongs to a unique quadrilateral. In other words, this means that when two vertices are adjacent, they share only one common neighbour (forming a triangle), and when two vertices are not adjacent, they share two common neighbours (forming a quadrilateral); therefore, the number of their common neighbours is determined by their adjacency. The Conway-99 problem asks for the existence of one such graph with 99 vertices and 14 regularity.

Brouwer and Haemers \cite{brouwer2011spectra} provide an equivalent definition of strongly regular graphs based on the field of spectral graph theory and linear algebra. A graph is strongly regular, if its adjacency matrix has exactly three eigenvalues. We will go through these in more detail in Sections \ref{preliminaries} and \ref{srgGroups}.

\section{Boolean satisfiability problem}
The boolean satisfiability problem (SAT) is one fundamental problem in mathematics and computer science. It involves determining the possibility of assigning truth values to boolean variables (0/\syntactic{False}) or (1/\syntactic{True}) so that the formula evaluates to \syntactic{True}. The Cook-Levin \cite{Cook71}\cite{Levin73} theorem concludes that SAT problem is NP-Complete, which means any NP problem can be reduced to an instance of the SAT problem. The simplicity of the SAT problem's form has motivated many to develop efficient algorithms to overcome this problem.

This is where SAT solver models come into play. They are designed to solve instances of the SAT problem. As an input, they receive a boolean formula and determine the existence of an assignment of truth values to the variables, such that the whole formula evaluates to \syntactic{True}. SAT solvers use a variety of algorithms and heuristics to reach and converge on a solution. 

It is trivial that the search space to find the Conway-99 graph is limited. However, an exhaustive search attack to this problem is beyond the capabilities of computing power. This led us to consider using SAT solvers to tackle the problem. SAT search is a widely used approach to go through the search space more efficiently.

A commonly used method in SAT solvers is \textit{symmetry-breaking}, which is particularly relevant when dealing with strongly regular graphs that have \textit{highly symmetrical} structures. This technique allows SAT solvers to avoid searching for equivalent cases, making them valuable tools in these scenarios. Even though strongly regular graphs exhibit symmetry, it is their distinctive asymmetrical characteristics that make them difficult to identify. More techniques and algorithms are commonly used, which we will discuss later in Section \ref{SATSearch}. Then, in Section \ref{experiments}, we will examine the experiments we had and analyse them mathematically.

\section{Work completed}
Our research involves converting the Conway-99 problem into a SAT problem to determine if it can be solved within a reasonable timeframe using SAT solvers. We generalised this using different SAT methods to search for any strongly regular graph with its specific properties. On this path, we studied the theoretical aspects of strongly regular graphs, their symmetries and patterns. We proved some of these patterns must not exist in any possible answers to the Conway-99.

Overall, the research combined two significant and intriguing areas of study that have the ability to provide meaningful insights into both fields.

\chapter{Preliminaries} \label{preliminaries}

\section{Basics}

In this work, for a graph $G = (V, E)$, the vertex set used is $\{1,2,\cdots,n\}$, while $n$ represents the number of vertices in a graph, or $n = |V|$. For edges, notations $\{u,v\}$ or $uv$ are used, while $u,v \in V$.  It is worth noting that all graphs under consideration are connected and simple, that is, undirected and without any double-edges or self-loops. Thus, when we refer to a graph, we specifically mean a simple graph. The neighbourhood of a vertex $v$ is denoted as $N(v)$, and its closed neighbourhood, that is, $N(v) \cup \{v\}$, is denoted as $N[v]$.

We begin by mathematically defining our main concept, strongly regular graphs:
\begin{definition}[Strongly regular graphs] \label{srgDefinition}
A graph $G = (V, E)$ is said to be strongly regular with parameters $(n, k, \lambda, \mu)$ if it satisfies the following conditions:
\begin{enumerate}
\item $G$ is regular of degree $k$ and is neither complete nor empty.
\item Every pair of adjacent vertices in $G$ has exactly $\lambda$ common neighbors.
\item Every pair of non-adjacent vertices in $G$ has exactly $\mu$ common neighbors.
\end{enumerate}
\end{definition}

In our context, when we state that a graph has a $\lambda$ ($\mu$) parameter, we mean every pair of adjacent (non-adjacent) vertices share $\lambda$ ($\mu$) neighbours. This is independent of the underlying graph being strongly regular or not.

While searching for strongly regular graphs, it can be readily proved that for a $(n, k, \lambda, \mu)$ strongly regular graph $G$, its complement $\overline{G}$ is a strongly regular graph with parameters $(n, \overline{k}, \overline{\lambda}, \overline{\mu})$:
\begin{align*}
    \overline{k} &= n - k - 1\\
    \overline{\lambda} &= n - 2 - 2k + \mu \\
    \overline{\mu} &= n - 2k - \lambda
\end{align*}

A graph $G$ is called a \textit{trivial} strongly regular, and is \textbf{not} considered as a strongly regular graph, if either itself or its complement $\overline{G}$ is disconnected. These graphs satisfy strong regularity conditions trivially. As an example, the complete $K_n$ and the empty graphs trivially satisfy these conditions.

Our focus in this section will be to limit the possible parameters further and examine the specific requirements and conditions that must be met by a set of parameters $(n,k,\lambda,\mu)$ so that there could exist an $(n,k,\lambda,\mu)$ strongly regular graph.

\begin{theorem} \label{basicDoubleCounting}
    In a strongly regular graph with parameters $(n,k,\lambda,\mu)$, the following equation holds:
    \begin{equation*}
        (n - k - 1)\mu = k(k - \lambda - 1).
    \end{equation*}
\end{theorem}

\begin{proof}
    To prove this, we will employ a double-counting approach and count the number of edges between two sets: $N(v)$ and $V - N[v]$. The size of the set $V - N[v]$ equals $n - k - 1$. Each vertex $u \in V - N[v]$ is connected to exactly $\mu$ vertices in $N(v)$. As a result, there are $(n - k - 1)\mu$ edges between these two sets.

    On the other hand, there exist $k$ vertices in $N(v)$, such that each of them has $\lambda$ neighbours within $N(v)$ and one neighbour $v$ itself. Therefore, $k - \lambda - 1$ edges remain for each vertex, connecting it to $v$'s non-neighbours. This gives us a total of $k(k - \lambda - 1)$ edges between the sets.

    By counting the same thing with two different expressions, we can infer that the equation
    \begin{equation} \label{basicSRGEquation}
        (n - k - 1)\mu = k(k - \lambda - 1)
    \end{equation}
    holds.
\end{proof}

\section{Algebraic graph theory}
The algebraic implications of strongly regular graph parameters are of significant interest to mathematicians and researchers in the field. The study of such parameters can provide valuable insights into the structural properties of graphs and their relationships with other mathematical concepts.

To study these properties, let $A$ be the graph $G$'s adjacency matrix. The identity matrix is denoted by $I$, while the all-one matrix is $J$.

\begin{theorem}\label{srgAlgTh}
    A graph $G$ is an $(n,k,\lambda,\mu)$ strongly regular graph if, and only if, its adjacency matrix satisfies the equation \begin{equation} \label{srgAlgEq}
        A^2 = kI + \lambda A + \mu(J - I - A).
    \end{equation}
\end{theorem}

There is a more generalised result for regular graphs that states a graph is connected and regular if and only if the matrix $J$ is a linear combination of powers of $A$. The proof of it is similar to the proof of Theorem \ref{srgAlgTh}. Their proofs rely on the following lemma regarding the powers of an adjacency matrix:

\begin{lemma}
    Let $G$ be a graph with an adjacency matrix $A$. For a pair of vertices $\{i,j\}$, and for any non-negative integer $p$, the $ij$-th entry of the matrix $A^p$ is the number of walks of length $p$ from vertex $i$ to vertex $j$ in $G$.
\end{lemma}
\begin{proof}
An induction on $p$ will be used.

\textbf{Base case}: For $p = 0$, the $ij$-th entry of $A^0 = I$, which is equal to the identity matrix, is 1 on its diagonal. This is in line with the fact that there is one walk of length 0 from any vertex to itself and no walks of length 0 between different vertices.

\textbf{Induction step}:
We assume the lemma holds for $k$. We prove it for $k + 1$.

For the $ij$-th entry of matrix $A^{k + 1}$ the equation
\begin{equation*}
    a^{k + 1}_{ij} = \sum_{v = 1}^n a^k_{iv}a_{vj}
\end{equation*} holds. This is the sum over the products of the $iv$-th entry of $A^k$ and the $vj$-th entry of $A$.

By induction's hypothesis, the $iv$-th entry of $A^k$ is the number of walks of length $k$ between $i$ and $v$. Therefore, if vertex $v$ is adjacent to $j$, extending all these walks with the edge $\{v,j\}$ would result in new walks that are of length $k + 1$. As a result, all walks are counted. Also, since walks of length $k - 1$ are counted only once, and the newly counted walks are different as they all end in different vertices $v$, each walk gets counted exactly once, and the lemma holds.
\end{proof}

\begin{proof}[Proof of theorem \ref{srgAlgTh}]
    The number of walks of length two between a vertex and itself equals its degree; hence, in a $k$-regular graph, matrix $A^2$'s diagonal equals $k$. On the other hand, if Equation \ref{srgAlgEq} holds, the graph must be $k$-regular.
    
    A pair of adjacent vertices ${u,v}$ in a strongly connected graph share $\lambda$ neighbours. We can infer that there are $\lambda$ walks of length two between $u$ and $v$, if, and only if, the equality $a^2_{u,v} = \lambda$ holds for all $uv$ such that $a_{u,v} = 1$.
    
    For non-adjacent vertices, the number of shared neighbours, and thus, the walks of length two in an $(n, k, \lambda, \mu)$ strongly regular graph, are $\mu$. The matrix $J - I - A$ contains $1$ only on disconnected pairs of vertices. Thus, the matrix $A^2$ contains $\mu$ where on indices where $J - I - A$ is $1$, if and only if, each two non-neighbours share $\mu$ neighbours.
    
    The summation of the multiplications of matrix $I$ by the regularity parameter $k$, matrix $A$ by $\lambda$ parameter, and matrix $J - I - A$ by $\mu$ parameter will generate the matrix $A^2$ if and only if $G$ is an $(n,k,\lambda,\mu)$ strongly regular graph.
\end{proof}

We examined the algebraic representation of the Definition \ref{srgDefinition} through the equation \ref{srgAlgEq}.

\begin{definition} [Eigenvalues and eigenvectors]
For an adjacency matrix $A$, a scalar $\tau$ is called an eigenvalue of $A$ if there exists a non-zero vector $x$, such that $Ax = \tau x$. The vector $x$ is called an eigenvector of the corresponding matrix. The multiplicity of eigenvalue $\tau$ is the number of linearly independent eigenvectors corresponding to it.
\end{definition}

Consider vector $\textbf{1} = ( 1, 1, \cdots, 1)^T$ of length $n$. If $G$ is a $k$-regular graph, the sum of each row and column of its adjacency matrix $A$ would be $k$. Consequently, the equation
\begin{equation} \label{kIsAnEigenvalue}
    A\textbf{1} = k\textbf{1}
\end{equation} holds. This implies that $k$ is an eigenvalue and $\textbf{1}$ is an eigenvector of the adjacency matrix $A$.

\begin{theorem} \label{kRegularEigenvalueMultiplicity}
    The maximum absolute eigenvalue of a connected graph $G$'s adjacency matrix $A$ is $\Delta(G)$. This maximum value is attained if, and only if, the graph is $\Delta(G)$-regular. The multiplicity of $\Delta(G)$ as an eigenvalue is one.
\end{theorem}
\begin{proof}
    Let $x$ be the eigenvector of matrix $A$ with the largest absolute value at coordinate $i$, and let $\tau$ be its eigenvalue. The inequality
    \begin{equation} \label{maxDeltaEigenvalue}
        |\tau||x_i| = |(Ax)_i| = \left|\sum a_{ij} x_j\right| \leq deg(v_i)|x_i| \leq \Delta(G) |x_i|
    \end{equation}
    holds, as there are $deg(v_i)$ non-zero entries on the $i$-th row of matrix A, with an absolute value less than $|x_i|$. Therefore, for each eigenvalue $\tau$, inequality $|\tau| \leq \Delta(G)$ holds.

    For all inequalities in Equation \ref{maxDeltaEigenvalue} to be equalities, all neighbours $v_j$ of vertex $v_i$, that is, $v_j \in N(v_i)$ and $a_{ij} = 1$, need $x_i = |x_j|$. The degree of $v_i$ must also be $\Delta(G)$.

    Since the graph is connected, and the neighbours of vertex $v_i$ all have the largest value in their respective coordinate in vector $x$, by iterating these equalities for more vertices, we infer that $x$ is a constant vector, and the graph is $\Delta(G)$-regular. Since vector $x$ is a constant vector, that is, a multiplication of vector \textbf{1}, eigenvalue $\Delta(G)$ has multiplicity one.
\end{proof}

The following lemma is useful, as an adjacency matrix is symmetric and has real values. Note that for a vector $v$ or a matrix $A$, its transpose is denoted as $v^T$ and $A^T$, respectively.

\begin{lemma} \label{orthogonalEigenvectors}
Let $A$ be a real symmetric matrix. If vectors $u$ and $v$ are eigenvectors of $A$, each corresponding to a distinct eigenvalue, then $u$ and $v$ are orthogonal.
\end{lemma}
\begin{proof}
The two eigenvectors must satisfy the following equations: $Au = \tau u$ and $ Av = \theta v$.
Since $A$ is symmetric, we can infer:
\begin{equation*}
    (Au)^T = u^TA^T = u^TA.
\end{equation*} Thus, by multiplying this with $v$, we can see that:
\begin{equation*}
    (u^T\tau) v = (u^TA)v = u^T(Av) = u^T(\theta v)
\end{equation*} holds; since $u^Tv(\tau - \theta) = 0$, and $\theta \neq \tau$, we can infer that $u$ and $v$ are orthogonal.
\end{proof}

\begin{theorem} \label{threeEigenvalues}
    The adjacency matrix $A$ of an $(n,k,\lambda,\mu)$ strongly regular graph has three distinct eigenvalues.
\end{theorem}
\begin{proof}
    Based on Theorem \ref{kRegularEigenvalueMultiplicity}, we know that $k$ is an eigenvalue of multiplicity one with eigenvector \textbf{1}.
    
    Consider Equation \ref{srgAlgEq}, and suppose $\theta$ is an eigenvalue of $A$, and its eigenvector is $x \neq \textbf{1}$. Rewrite the equation by multiplying it with the eigenvector $x$, as follows:
\begin{equation*}
    A^2x = kx + \lambda Ax + \mu(J - I - A)x.
\end{equation*}
By considering $Ax = \theta x$, 
\begin{equation*}
    \theta^2x = kx + \lambda \theta x + \mu Jx -\mu x - \mu Ax.
\end{equation*}
According to Lemma \ref{orthogonalEigenvectors}, columns of matrix $J$ are orthogonal with $x$; thus, $\mu Jx = 0$ and the equation can be rewritten as:
\begin{equation*}
    \theta^2x = kx + \lambda \theta x - \mu x - \mu \theta x.
\end{equation*}
Ergo, the eigenvalues of matrix $A$ must satisfy the equation
\begin{equation*}
    \theta^2 + (\mu - \lambda)\theta + (\mu - k) = 0.
\end{equation*}

Two more eigenvalues can be obtained by this quadratic equation:
\begin{equation} \label{eigenvaluesR}
r = \frac{1}{2} \Big[ (\lambda - \mu) + \sqrt{(\lambda - \mu)^2 + 4(k - \mu)}\Big],
\end{equation}
and
\begin{equation} \label{eigenvaluesS}
s = \frac{1}{2} \Big[ (\lambda - \mu) - \sqrt{(\lambda - \mu)^2 + 4(k - \mu)}\Big];
\end{equation} therefore, we have exactly three eigenvalues for the adjacency matrix of a strongly regular graph. Considering values $r$ and $s$, since 
\begin{equation*}
    \sqrt{(\lambda - \mu)^2 + 4(k - \mu)} > \lambda - \mu,
\end{equation*} 
eigenvalue $r$ is positive, and $s$ is negative.

\end{proof}

\begin{definition} [Characteristic polynomial]\label{characteristicPoly}
    The characteristic polynomial of a matrix $A$ is the polynomial 
    \begin{equation*}
        P(\tau) = \text{det}(\tau I - A).
    \end{equation*} 
    
\end{definition}
The eigenvalues of a matrix are the roots of the characteristic polynomial, and this polynomial is of degree $n$.

Consider an eigenvalue $\tau$ and an eigenvector $x$ for matrix $A$: $Ax = \tau x$. Rewriting the equation as $x(\tau I - A) = 0$, we can infer that $(\tau I - A)$ has no inverse. Consequently, det$(\tau I - A) = 0$, which provides a polynomial equation whose roots correspond to the eigenvalues of $A$. The number of times an eigenvalue appears as a root is equal to its multiplicity.

\begin{definition}
    The trace of an $n \times n$ square matrix $A$, denoted as $tr(A)$, is defined as:
    \begin{equation*}
        tr(A) = \sum_{i=1}^{n} a_{ii} = a_{11} + a_{22} + a_{33} + \ldots + a_{nn}.
    \end{equation*}
\end{definition}
\begin{lemma} \label{traceEigenvalue}
    The sum of diagonal entries of a matrix $A$, i.e. its trace $tr(A)$, is equal to the summation of $A$'s eigenvalues:
    \begin{equation*}
        tr(A) = \sum_{i=1}^{n} \tau_i.
    \end{equation*}
\end{lemma}
\begin{proof}
    The characteristic polynomial of $A$ is given by:
    \begin{equation*}
        p(\tau) = \text{det}(\tau I - A),
    \end{equation*}
    which can be rewritten using eigenvalues of $A$ as
    \begin{equation*}
        p(\tau) = (\tau - \tau_1)(\tau - \tau_2)\ldots(\tau - \tau_n).
    \end{equation*} Expanding this product 
    \begin{equation*}
        p(\tau) = \tau^n - (\tau_1 + \tau_2 + \ldots + \tau_n )\tau^{n - 1} + \ldots
    \end{equation*}
    shows us that the coefficient of $-\tau^{n-1}$ is the negation of the summation of eigenvalues. 

    The Leibniz formula is another way to calculate the determinants of matrices. For a matrix $Z$, the formula is defined as follows:
    \begin{equation*}
        det(Z) = \sum_{p \in S_n} sgn(p) \prod_{i=1}^{n} z_{i,p(i)},
    \end{equation*}
    where $z_{i,p(i)}$ is $Z$'s $i$-th row and $p(i)$-th column element; the group $S_n$ is the symmetric group defined over a finite set of $n$ symbols. The order of the group $S_n$ is $n!$ as it contains all permutations of $n$ elements. Ergo, each $p$ is a permutation. The sign function $sgn$ is defined to be +1 if the parity of a permutation is even, and -1 for odd parity. The parity of a permutation $p$ is the parity of its number of inversions, that is, two values $a$ and $b$ such that $a < b$ and $p(a) > p(b)$.

    Calculating the determinant of a matrix with this approach requires $O(n!)$ operations, which makes it impractical. However, it provides us with more insights into understanding the determinants of matrices.

    Considering the matrix $\tau I - A$, the only permutation capable of producing $\tau^{n - 1}$ is the identity, as otherwise, at least two non-diagonal elements would be chosen, and the maximum order of $\tau$ can be $n - 2$. For the identity permutation, the polynomial generated is 
    \begin{equation*}
        \prod_{i=1}^{n}(\tau - a_{ii}).
    \end{equation*}
    Hence, the coefficient of $\tau^{n - 1}$ is $-\sum_{i=1}^{n} a_{ii}$. Therefore, the equation
    \begin{equation*}
        tr(A) = \sum_{i=1}^{n} a_{ii} = \sum_{i=1}^{n} \tau_i
    \end{equation*} holds.
\end{proof}

Lemma \ref{traceEigenvalue} comes in handy when working with graphs and adjacency matrices, as the trace of simple graphs are 0. In Theorem \ref{threeEigenvalues}, we proved that exactly three distinct eigenvalues exist. Let $m_r$ and $m_s$ each represent the multiplicity (number of times the eigenvalue appears as a root of the characteristic polynomial) of their respective eigenvalue. According to Lemma \ref{traceEigenvalue}, we have
\begin{equation*}
    m_rr + m_ss + k = tr(A) = 0,
\end{equation*}
and given that we are dealing with a polynomial of degree $n$, it follows that:
\begin{equation*}
    m_r + m_s = n - 1.
\end{equation*}
Hence, we can calculate the multiplicities based on their eigenvalues:
\begin{align} \label{eigenvalueMultiplicities}
&m_s = \dfrac{(n-1)r + k}{r - s} , & m_r = \dfrac{(n-1)s + k}{s - r}.&
\end{align}
In Equations \ref{eigenvaluesR} and \ref{eigenvaluesS}, the eigenvalues were expressed in terms of the strongly regular graph's parameters, $(n,k,\lambda,\mu)$. Substituting the values based on it yields
\begin{equation} \label{rEigenvalueMultiplicity}
    m_r = \frac{1}{2}\Big[ (v - 1) - \frac{2k + (v-1)(\lambda - \mu)}{ \sqrt{(\lambda - \mu)^2 + 4(k - \mu)}} \Big],
\end{equation}
and
\begin{equation} \label{sEigenvalueMultiplicity}
    m_s = \frac{1}{2}\Big[ (v - 1) + \frac{2k + (v-1)(\lambda - \mu)}{ \sqrt{(\lambda - \mu)^2 + 4(k - \mu)}} \Big].
\end{equation}
A powerful condition, the \textit{integrality condition}, can be derived from the multiplicities, as these numbers must be natural numbers. Therefore, many strongly regular graph parameter sets can be ruled out, as they will not yield natural numbers for the multiplicities. We will use this further in Section \ref{srgGroups} to narrow down the parameter set possibilities based on our needs.

\begin{lemma} \label{linearCombinationPowerTwo}
Let $A$ denote a simple and connected graph's adjacency matrix with $n$ vertices. If $A^2$ is a linear combination of three matrices $A$, $I$, and $J$, then the graph it represents is strongly regular.
\end{lemma}
\begin{proof}
    Let $a$, $b$, and $c$ denote the coefficients in the linear representation. Then,
    \begin{equation*}
        A^2 = aI + bA + cJ,
    \end{equation*}
    so
    \begin{equation*}
        A^2 = (a + c)I + (b + c)A + c(J - I - A).
    \end{equation*}
    According to Theorem \ref{srgAlgTh}, it can be inferred that this represents a strongly regular graph with parameters $(n, a + c, b + c, c)$.
\end{proof}

\begin{theorem}
    A $k$-regular connected graph $G$ with an adjacency matrix $A$ that has exactly three distinct eigenvalues is strongly regular.
\end{theorem}
\begin{proof}
    Since the graph is $k$-regular, one of these eigenvalues is $k$. Name the other two eigenvalues as $r$ and $s$. 

    Let $\tau$ denote an eigenvalue of matrix $(A - rI)(A - sI)$, and $x$ be an associated eigenvector:
    \begin{equation*}
        (A - rI)(A - sI) \cdot x = \tau x.
    \end{equation*}
    Since matrix $A$ has only three eigenvalues, and the sum of $r$ and $s$'s multiplicities is $n - 1$, matrix $(A - rI)(A - sI)$ has $0$ as an eigenvalue of multiplicity $n - 1$.

    The eigenvalues of matrix $A - rI$ are obtained by the equation
    \begin{equation*}
        (A - rI) x = \tau x.
    \end{equation*}
    By rewriting the equation as follows: 
    \begin{equation*}
        Ax = (r + \tau)x,
    \end{equation*}
    we can determine that $r + \tau$ is an eigenvalue of the initial matrix $A$. Moreover, we can deduce that $(k - r)$ is an eigenvalue of matrix $(A - rI)$. With the same approach, we can see $(k - s)$ is an eigenvalue for $(A - sI)$. In Equation \ref{kIsAnEigenvalue}, we saw that the eigenvector $\textbf{1}$ belongs to the eigenvalue $k$. Therefore, 
    \begin{equation*}
        (A - rI)(A - sI)\textbf{1} = (A - rI) (k-s)\textbf{1} = (k-r)(k-s)\textbf{1}
    \end{equation*}
    holds, and the only non-zero eigenvalue is $(k - r)(k - s)$ with multiplicity 1 and eigenvector $\textbf{1} = (1,1, \ldots, 1)^T$. Because there is only one linearly independent eigenvector, \textbf{1}, all the entries in the matrix have the same value. Thus,
    \begin{equation*}
        (A - rI)(A - sI) = \frac{(k-r)(k-s)}{n}J.
    \end{equation*}
    By expanding the left side of the equation, we obtain
    \begin{equation*}
         A^2 - (r+s)A + rsI = \frac{(k-r)(k-s)}{n}J.
    \end{equation*}
    Based on Lemma \ref{linearCombinationPowerTwo}, it is concluded that $A$ is the adjacency matrix of a strongly regular graph.
\end{proof}

\section{Automorphism groups}
We begin with two graph-related definitions:
\begin{definition}
    An \textbf{isomorphism} between two graphs $G$ and $H$ is defined by a bijective function
    \begin{equation*}
        f: V(G) \rightarrow V(H)
    \end{equation*}
    such that $\{u,v\} \in E(G)$ if and only if $\{f(u),f(v)\} \in E(H)$.
\end{definition}
\begin{definition}
    An \textbf{automorphism} refers to an isomorphism that maps a graph $G$ to itself. It is defined using a permutation $\varphi$, such that $\{u,v\} \in E(G)$ if and only if $\{\varphi(u),\varphi(v)\} \in E(G)$.
\end{definition}

In group theory, a group is defined as follows:
\begin{definition}
    A \textbf{group} $G$ is a non-empty set with a binary operation `.'
    \begin{equation*}
        \cdot: G \times G \rightarrow G
    \end{equation*}
    such that it satisfies three requirements:
    \begin{enumerate}
        \item Associativity: for all $a$,$b$,$c$ in $G$, equality $(a\cdot b)\cdot c = a\cdot (b\cdot c)$ holds.
        \item Identity: there exists $i \in G$, such that for all $a \in G$, equalities $a = a\cdot i = i\cdot a$ hold.
        \item Inversion: for all $a \in G$, there exists an element $a^{-1} \in G$, such that $a\cdot a^{-1} = a^{-1}\cdot a = i$.
    \end{enumerate}
\end{definition}

Consider all the permutations that form automorphisms in a graph $G$. The composition of any two automorphisms $\varphi_1$ and $\varphi_2$ yields an automorphism $\varphi_1\varphi_2$. All three group requirements are satisfied in this set:
\begin{enumerate}
    \item Associativity: for three automorphisms $\varphi_1$, $\varphi_2$, and $\varphi_3$, the equality $\varphi_1\cdot (\varphi_2\cdot \varphi_3(x)) = (\varphi_1\cdot \varphi_2)\cdot \varphi_3(x)$ holds.
    \item Identity: the identity permutation, denoted by $\varphi_I$ is an automorphism.
    \item Inversion: If $\varphi$ is an automorphism, then its inverse permutation $\varphi^{-1}$ is also an automorphism.
\end{enumerate}
Hence, the set of all automorphisms is a group, which is named the automorphism group and is denoted with $Aut(G)$.

Automorphism groups are meaningful in studying strongly regular graphs because their order indicates the level of \textit{symmetry} in graphs. The order of $Aut(K_n)$ for a complete graph $K_n$ is $n!$, making it highly symmetrical. On the other hand, a graph $G$ is considered asymmetric if its automorphism order, denoted as $|Aut(G)|$, is one. It is noteworthy that numerous strongly regular graphs do not possess high symmetries. We will review this later in Section \ref{experiments}.

\begin{definition}[Spectrum]
    The spectrum of a matrix is the set of its eigenvalues and their multiplicity. In algebraic graph theory, a \textbf{graph's spectrum} is the graph's adjacency matrix's spectrum.
\end{definition}
Consider the following as an example of a strongly regular graph's spectrum:
\begin{example}
    The Petersen graph, a $(10,3,0,1)$ strongly regular graph, has three eigenvalues: $k = 3$, $r = 1$, and $s = -2$. The multiplicity of the eigenvalues are $m_k = 1$, $m_r = 5$, and $m_s = 4$. To denote the spectrum, we write
    $$
    \tau(A) = \{3^1, 1^5, (-2)^4\}.
    $$
\end{example}
The Petersen graph is an \textbf{\textit{integral}} graph. A graph is integral if its spectrum consists of integers. When two graphs have the same spectrum, they are referred to as \textit{cospectral}. Although isomorphic graphs are cospectral, the reverse is not always true. This is evident even among the highly structured, strongly regular graphs, which we will explore in Section \ref{experiments}.

\begin{table}[htbp]
\centering
\renewcommand{\arraystretch}{1.2}
\begin{tabular}{rrrr|crcr}
\hline
$n$ & $k$ & $\lambda$ & $\mu$ & $r$ & $m_r$ & $s$ & $m_s$  \\
\hline
5 & 2 & 0 & 1 & $(-1 + \sqrt{5})/2$ & 2 & $(-1 - \sqrt{5})/2$ & 2\\
9 & 4 & 1 & 2 & 1 & 4 & -2 & 4\\
10 & 3 & 0 & 1 & 1 & 5 & -2 & 4\\
13 & 6 & 2 & 3 & $(-1 + \sqrt{13})/2$ & 6 & $(-1 - \sqrt{13})/2$ & 6\\
15 & 6 & 1 & 3 & 1 & 9 & –3 & 5\\
16 & 5 & 0 & 2 & 1 & 10 & -3 & 5\\
16 & 6 & 2 & 2 & 2 & 6 & -2 & 9\\
17 & 8 & 3 & 4 & $(-1 + \sqrt{17})/2$ & 8 & $(-1 + \sqrt{17})/2$ & 8\\
21 & 10 & 5 & 4 & 3 & 6 & -2 & 14\\
\hline
\end{tabular}
\caption{Parameters of multiple strongly regular graphs and their spectrums. The eigenvalue $k$ and its multiplicity 1 is omitted.}
\label{smallSRGs}
\end{table}

\chapter{Strongly regular graphs} \label{srgGroups}
In this chapter, we will be studying the theoretical aspects of several classes of strongly regular graphs and how they are constructed. Moreover, we apply the conditions we have learned in the Preliminaries chapter \ref{preliminaries}.

We begin this chapter by mathematically defining the Conway-99 open problem. Next, we will examine several similar structures to the potential Conway-99 graph. Finally, we will analyse these structures and their connection to the Conway-99 problem.

\section{The Conway-99 problem and similar graphs}
The Conway-99 is an open graph problem that asks the following question:
\begin{problem}[Conway-99]
    Does an instance of $(99,14,1,2)$ strongly regular graph exist?
    \label{conway99}
\end{problem}

The $\lambda = 1$ and $\mu = 2$ conditions state that every edge should be part of a \textit{unique triangle} and every non-edge should be part of a \textit{unique quadrilateral}. Only five possible parameter set exists for strongly regular Graphs with $\lambda = 1$ and $\mu = 2$. These parameters are shown in Table \ref{lambda1mu2Table}.

\begin{table}[htbp]
\centering
\begin{tabular}{p{0.05\textwidth}| p{0.10\textwidth} p{0.08\textwidth} p{0.05\textwidth} p{0.05\textwidth}}

\hline
$\exists$ & $n$ & $k$ & $\lambda$ & $\mu$ \\
\hline
+ & 9 & 4 & 1 & 2 \\
? & 99 & 14 & 1 & 2 \\
+ & 243 & 22 & 1 & 2 \\
? & 6273 & 112 & 1 & 2 \\
? & 494019 & 994 & 1 & 2 \\
\hline
\end{tabular}
\caption{Strongly regular graphs with $\lambda = 1$ and $\mu = 2$. Only the existence of two graphs has been verified.}
\label{lambda1mu2Table}
\end{table}

An instance of $(9,4,1,2)$ strongly regular graph is named the Paley(9) graph. Paley graphs form an infinite group of strongly regular graphs, which will be studied further in Section \ref{paleyGraphs}. 

In Section \ref{srg243}, we delve into the construction process of the Berlekamp-Van Lint-Seidel graph, a $(243,22,1,2)$ strongly regular graph, which was discovered in 1973.

To prove only five possible parameter sets exist, we use the multiplicity integrality condition and prove the following theorem:
\begin{theorem} \label{refRemoveParameter}
There are no strongly regular graphs with $\lambda = 1$ and $\mu = 2$ out of the parameter sets of Table \ref{lambda1mu2Table}.
\end{theorem}
\begin{proof}
    Let $(n,k,1,2)$ be a possible parameter set. By using the Equation \ref{basicSRGEquation} we obtain 
    \begin{equation*}
        (n - k - 1)2= k(k - 2),
    \end{equation*}
    so we can infer that
    \begin{equation*}
        n = \frac{k^2}{2} + 1
    \end{equation*}
    must hold.

    We can now rewrite the Equations \ref{rEigenvalueMultiplicity} and \ref{sEigenvalueMultiplicity}, the eigenvalues' multiplicities, only based on one parameter. We then obtain
    \begin{equation*}
        M = \frac{1}{2}\left[ \frac{k^2}{2} + \frac{2k - \frac{k^2}{2}}{\sqrt{4k - 7}}
        \right].
    \end{equation*}
    Since these multiplicities must be positive integers, we can rule out many possibilities for strongly regular graphs with $\lambda = 1$ and $\mu = 2$.

    To ensure the multiplicities are natural numbers, equality $4k - 7 = t^2$ must hold for a natural $t$. Replacing $k$ with $\frac{t^2 + 7}{4}$ we obtain
    \begin{equation*}
        M = \frac{1}{2}\left[ \frac{\left(\frac{t^2 + 7}{4}\right)^2}{2} + \frac{2\left(\frac{t^2 + 7}{4}\right) - \frac{\left(\frac{t^2 + 7}{4}\right)^2}{2}}{\sqrt{4\left(\frac{t^2 + 7}{4}\right) - 7}}
        \right].
    \end{equation*}

    To simplify the equation, we take these steps
    \begin{align*}
        M &= \frac{1}{2}\left[ \frac{\left(t^2 + 7\right)^2}{32} + \frac{\left(t^2 + 7\right) - \left(\frac{t^2 + 7}{4}\right)^2}{2t}
        \right]\\
        \Rightarrow M &= \frac{(7 + t^2) (9 - t^2 + 7 t + t^3)}{64t}\\
        \Rightarrow 64M &= \frac{(7 + t^2) (9 - t^2 + 7 t + t^3)}{t}\\
        \Rightarrow 64M &= \frac{63 - 7t^2 + 49t + 7t^3 + 9t^2 - t^4 + 7t^3 + t^5}{t}\\
        \Rightarrow 64M &= 49 + \frac{63}{t} + 2t + 14t^2 - t^3 + t^4
    \end{align*}
    The values that would make $\frac{63}{t}$ an integer would be $t \in \{\pm 1,\pm 3,\pm 7,\pm 9,\pm 21,\pm 63\}$. These would result in $k$ values $\{2, 4, 14, 22, 112, 994\}$. Since for $k = 2$, the resulting graph would be a complete graph $K_3$, and complete graphs are not considered as strongly regular graphs in their nature, this case would be omitted. The five final possibilities have regularity $k \in \{4, 14, 22, 112, 994\}$ and number of vertices $n \in \{9,99,243,6273,494019\}$, as shown in Table \ref{lambda1mu2Table}.
\end{proof}

\section{Paley graphs} \label{paleyGraphs}
Paley strongly regular graphs are constructed based on Galois fields. To begin, we must first define \textbf{\textit{field}}. A field is a set with four operations of arithmetic: addition, subtraction, multiplication, and division (excluding division by zero). For instance, the set of rational numbers $\mathbb{Q}$ is considered a field, while natural numbers do not form a field because not all divisions can be defined. For example, $\frac{3}{4} \notin \mathbb{N}$. The mathematical definition is as follows:
\begin{definition}[Field]
    A field is a set $F$ equipped with two binary operations, namely addition (+) and multiplication (+). These binary operations are mappings in the form $F \times F \rightarrow F$. The following axioms must hold in these operations:
    \begin{enumerate}
        \item Associativity: For all $a, b, c \in F$, we have $(a + b) + c = a + (b + c)$ and $(a \cdot b) \cdot c = a \cdot (b \cdot c)$.
        \item Commutativity: For all $a, b \in F$, we have $a + b = b + a$ and $a \cdot b = b \cdot a$.
        \item Identity Elements: There exist two distinct elements 0 and 1 in $F$ such that for all $a \in F$, we have $a + 0 = a$ and $a \cdot 1 = a$.
        \item Additive Inverses: For every $a \in F$, there exists an element $-a$ in $F$ such that $a + (-a) = 0$.
        \item Multiplicative Inverses: For every $a \neq 0$ in $F$, there exists an element $a^{-1}$ in $F$ such that $a \cdot a^{-1} = 1$.
        \item Distributivity: For all $a, b, c \in F$, we have $a \cdot (b + c) = a \cdot b + a \cdot c$.

    \end{enumerate}
\end{definition}

To better understand the concept, we consider this simple example. For the set $\{0,1,2\}$, by defining the aforementioned axioms, modulo the prime number 3, we obtain a field. The first two conditions arbitrarily hold. The identity elements are 0 and 1. The additive inverse of 1 is 2 and vice-versa. The multiplicative inverse of 2 is 2 itself, as $4 \equiv 1 \pmod 3$. This field is named $GF(3)$. Galois fields are defined as follows:
\begin{definition} [Galois field]
    A \textit{finite or Galois field} is a field that contains a finite number of elements.
\end{definition}
The order of a finite field must be a prime power, i.e. $p^n$. For each power, there exists, up to isomorphism, exactly one finite field, which we denote as $GF(p^n)$.

A \textit{quadratic residue} of a field $GF(p^n)$ is an element's square in $GF(p^n)$. For example, element 1 is the quadratic residue of 1 and 2 in $GF(3)$, since $1^2 \equiv 2^2 \equiv 1 \pmod{3}$, while 2 is not.

The set of quadratic residues of a field $GF(p^n)$ is denoted by $QR(p^n)$. That is,
\begin{equation*}
    a \in QR(p^n) \Leftrightarrow \exists b \in GF(p^n): b \cdot b = a.
\end{equation*}

The mathematical definition of Paley graphs' structure is described as follows:
\begin{definition}[Paley graphs]
    Let $q$ be a prime power of $p$, such that $q = p^n \equiv 1 \pmod{4}$. The graph $Paley(q) = (V,E)$ is defined over the set of vertices $V = GF(q)$.
    The set of edges $E$ is defined as
    \begin{equation*}
        E = \{\{u,v\} : u - v \in QR(q) \},
    \end{equation*} that is, two vertices form an edge if their values differ by a quadratic residue.
\end{definition}
 
We first build a simple example of Paley graphs, specifically, Paley(13). The field $GF(13)$ is defined as the usual arithmetic operations modulo 13, over the set of $\{0,1,2, \ldots, 12\}$. To obtain the quadratic residue of $GF(13)$, consider the squares of $\pm 1, \pm 2, \pm 3, \pm 4, \pm 5,$ and $ \pm 6$, which respectively give us $ +1, +4, -4, +3, -1$, and $-3$. In figure \ref{paley13}, an instance of Paley(13) is drawn. A vertex $v$ is adjacent to 6 more vertices $v \pm 1$, $v \pm 3$, and $v \pm 4$. 

To build the Paley(9) graph, which is a strongly regular graph with parameters $(9,4,1,2)$, we first need to build the field $GF(3^2)$. The process of building Galois fields $GF(p^n)$ for $n > 1$ is more complex. The process requires polynomials and, more specifically, irreducible polynomials. These polynomials are chosen so that the field's arithmetic structure remains solid. To start off, we need to find a polynomial of degree 2 that has no solutions in the field $GF(3)$. The polynomial $x^2 + 1$ has no zeroes in this field, and thus, it can be used here. We suppose there is an $x$, not in this field but in larger fields, that satisfies the equality $x^2 + 1 = 0$ for now. The set built around $x$ and $GF(3)$,
\begin{equation*}
    \{0, 1, 2, x, x + 1, x + 2, 2x, 2x + 1, 2x + 2\},
\end{equation*}
contains 9 elements; this set, based on the degree 2 polynomial $x^2 + 1$, satisfy the field's arithmetic conditions. For example, the additive inverse of $2x + 1$ is $x + 2$, as $3x + 3 = 0$; the multiplicative inverse of it is $2x + 2$ because
\begin{equation*}
    (2x + 1)(2x + 2) = 4x^2 + 6x + 2 = x^2 + 2 = 1.
\end{equation*}
The quadratic residue of $GF(9)$ is
\begin{equation*}
    \{1,2, x, 2x\}, 
\end{equation*} where 2 and $2x$ can be also written as $-1$ and $-x$. For an arbitrary vertex $v \in GF(9)$, by connecting it to four vertices $v \pm 1$ and $v \pm x$, modulo 3, we obtain the Paley(9) graph in Figure \ref{paley9xAnd1}. Note that we can also consider vertices as a pair $(i,j)$, where $i$ represents variable $x$'s coefficient.

\begin{figure}[h!]
  \centering
  \includegraphics[width=0.6\linewidth]{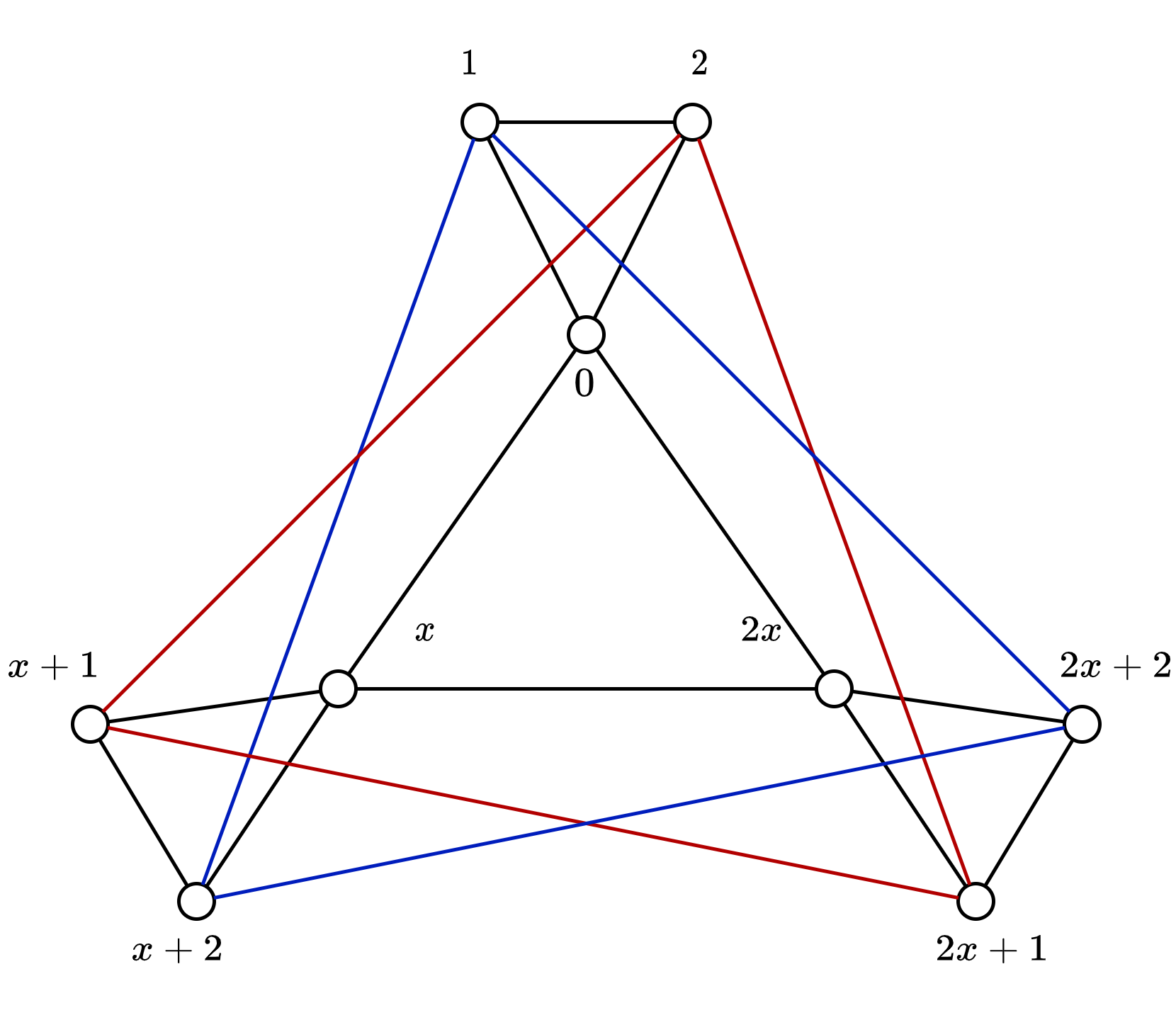}
  \caption{Paley(9) with $GF(9)$ vertices. This representation highlights the triangles.}
  \label{paley9xAnd1}
\end{figure}

In a Galois field of order $q$, the size of the quadratic residues' set $QR(q)$ is $\frac{1}{2}(q - 1)$. Every two elements of $GF(q)$ share $\frac{1}{4}(q - 1)$ neighbours unless the two elements differ by one quadratic residue. In that case, they share $\frac{1}{4}(q - 5)$ neighbours, which is one less because of the quadratic residue connecting them. Therefore, for every prime power $q \equiv 1 \pmod 4$, there exists a Paley(q) graph with the parameter set:
\begin{equation*}
    (q, \frac{1}{2}(q - 1), \frac{1}{4}(q - 5), \frac{1}{4}(q - 1)).
\end{equation*}
In the next section, we will use Paley(9) graph to build a larger, strongly regular graph with $\lambda = 1$ and $\mu = 2$.

\section{Berlekamp–Van Lint–Seidel graph} \label{srg243}
The Berlekamp-Van Lint–Seidel is the largest strongly regular graph discovered with $\lambda = 1$ and $\mu = 2$. This graph was constructed using perfect ternary Golay codes \cite{berlekamp1973strongly}, which have their roots in the field of data transmission and coding theory. This example highlights the overlap between the strongly regular graphs and coding theory and shows that graphs can be used to generate better coding schemes, depending on the need to transmit data. To gain a better understanding of how this graph was created, we must first complete some preliminary steps.

\subsection{Preliminaries}
While transmitting data, usually binary data containing bits or ternary data containing \textit{trits}, over a noisy channel susceptible to interference, the transmitted bit may experience alteration. The data $X$ being transmitted over a channel can be affected by a noise $N$ (note that the signal being sent is not a discrete value), and the result would be $Y = X + N$, which might be different than expected. Some coding schemes in coding theory are used to reduce the chance of data loss.

As an example, consider a 3-bit message code $(c_1,c_2,c_3)$ and a parity check matrix which we can use to generate

\[
\begin{bmatrix}
c_4 \\
c_5 \\
c_6 \\
\end{bmatrix} = 
\begin{bmatrix}
1 & 1 & 0 \\
1 & 0 & 1 \\
0 & 1 & 1 \\
\end{bmatrix}
\begin{bmatrix}
c_1 \\
c_2 \\
c_3 \\
\end{bmatrix}
\] three more parity check bits. After receiving the data, if the bits do not satisfy the parity conditions, we can infer an error has occurred. The word $(c_1,c_2,c_3,c_4,c_5,c_6)$ is considered as a codeword. A \textit{codeword} is a word that is generated by the transmitter; thus, it has not been affected by any noise and is correct.

Within the field of information theory, error-correcting codes serve a dual purpose: not only do they help the detection of errors, but they also enable their correction. To do this, they encode the data into a larger space, and by increasing the distance between the codewords, they offer a high chance of correcting the data \ref{error_correction}. The distance measure used is the Hamming distance between codewords: the number of positions (indices) in which they differ. The number of elements in a codeword that are not zero determines its weight.

\begin{figure}
    \centering
    \includegraphics[width=0.7\textwidth]{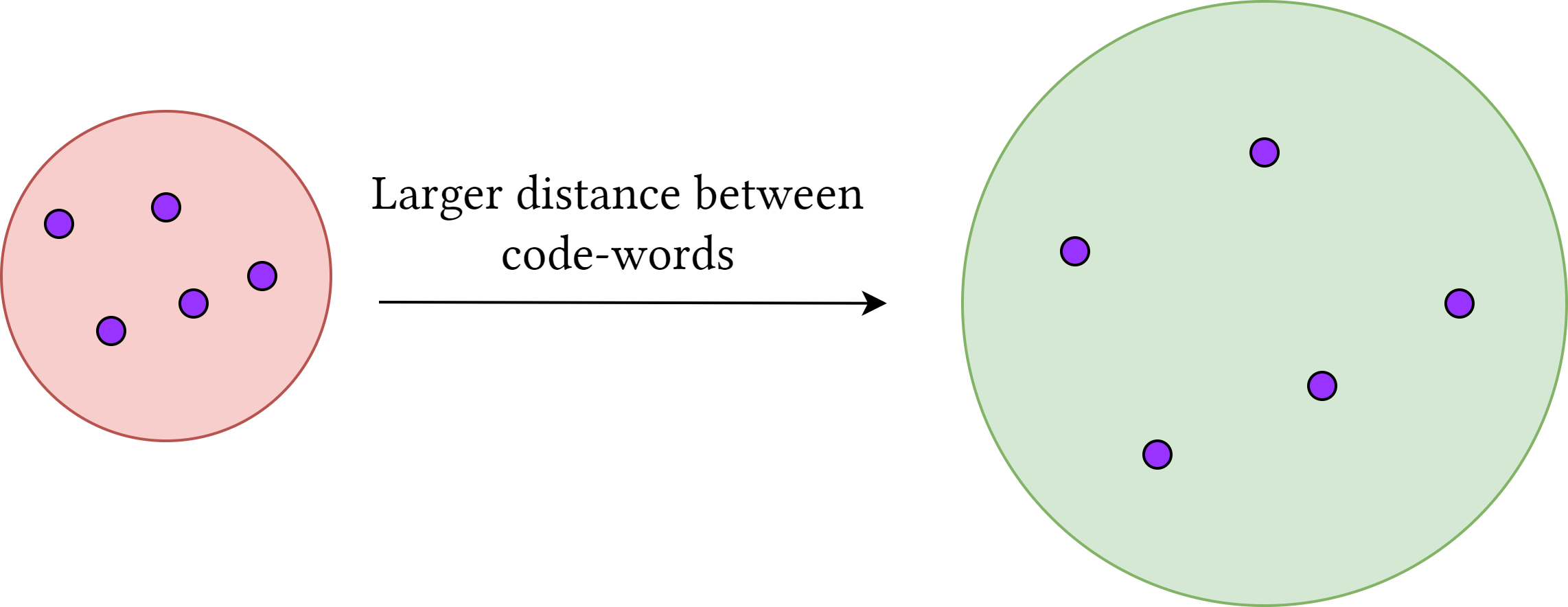}
    \caption{Correcting errors by having larger distances between code words}
    \label{error_correction}
\end{figure}

\subsection{Ternary Golay codes}

The ternary Golay code \cite{golay1949notes}, discovered by Marcel J. E. Golay, is one of the remarkable examples of \textit{perfect} codes. Perfect codes are the codes with maximum efficiency, that is, by utilising their redundancy efficiently, these codes enable maximum error correction. In mathematical terms, a code can be considered perfect if the spheres - consisting of words with a maximum Hamming distance that are centred around the codewords - cover the entire space, and every word belongs to the sphere of a unique codeword.

With the ternary Golay code, 6 trits of information can be encoded into 11 trits of data. This coding scheme allows for the correction of up to 2 errors. This can also mean that the distance between each pair of codewords is at least 5; otherwise, there would be a word that is at a distance of 2 or less of multiple codewords.

This code is defined over the finite field of $GF(3)$, which consists of the elements \{0,1,2\} or \{-1,0,1\}. The ternary Golay code can be generated using a generator matrix, denoted by the symbol $G$, which is equal to:
\begin{align*}
    G = \left[ I_6\ \ |\ \ P\right] = &\left[ 
    \begin{array}{cccccc|ccccc}
    1 & 0 & 0 & 0 & 0 & 0 & 1 & 1 & 1 & 1 & 1 \\
    0 & 1 & 0 & 0 & 0 & 0 & 1 & 1 & 2 & 2 & 0 \\
    0 & 0 & 1 & 0 & 0 & 0 & 1 & 2 & 1 & 0 & 2 \\
    0 & 0 & 0 & 1 & 0 & 0 & 2 & 1 & 0 & 1 & 2 \\
    0 & 0 & 0 & 0 & 1 & 0 & 2 & 0 & 1 & 2 & 1 \\
    0 & 0 & 0 & 0 & 0 & 1 & 0 & 2 & 2 & 1 & 1 \\
    \end{array}
    \right].
\end{align*}
Notice that the structure of $G$ allows a distance of 5 between codewords. A word $w$ can be represented as a vector with 6 dimensions, where each coordinate is 0,1 or 2.

A \textit{vector space} over a field (in this case, the Galois field GF(3)) is a set of vectors that must satisfy multiple axioms. These axioms 
include, for vectors $u$ and $v$ and scalar $a$, associativity ($(u + v) + w = u + (v + w)$), commutativity ($u + v = v + u$), identity element ($u + 0 = u$), inverse elements ($u + -u = 0$), and distributivity of scalar ($a(u + v) = au + av$)\footnote{We have omitted some axioms because of their simplicity, but we still believe that the remaining ones convey the general idea.}. We denote the vector space of dimension $l$ on $GF(q)$ with $V(l,q)$. As an example, consider vector $(0,0,2)$ of length three in $GF(3)$, so $(0,0,2) \in V(3,3)$.

It is worth noting that this code is classified as linear because of its generation process. That is, for a word $w \in V(6,3)$, the codeword $w\cdot G \in V(11,3)$ is generated by multiplication.

\subsection{Construction of the Berlekamp–Van Lint–Seidel graph}
We show two ways of constructing this graph: partitioning and parity.

\subsubsection{The partitioning approach}
The perfect ternary Golay code can be used to partition the vector space $V(11,3)$ into $3^5$ \textit{cosets} by adding $3^5$ different fixed words to all codewords. A subset of vectors $H$ from a vector space $V$ is a subspace if itself is a vector space; that is, it satisfies the axioms we defined previously under the same addition and scalar multiplications operators. Considering a vector $v \in V$, the set of vectors $H_v = \{v + h: h\in H\}$ is a \textit{coset}, with $v$ as its \textit{coset leader}. The coset leader refers to the word with the lowest weight within a coset.

Let $C$ represent the set of codewords. Considering how the codewords were linearly generated, it is not hard to see that this set forms a subspace of the vector space $V(11,3)$. Since the ternary Golay code encodes vectors of length 6, we have $3^6$ codewords.

Let $W_2$ be the set of vectors whose weight does not exceed 2. The size of $W_2$ is
\begin{equation*}
    1 + 2 \times 11 + 4 \times \binom{11}{2} = 243.
\end{equation*}
This count represents the exact number of vertices we seek for.

We can use the vectors in $W_2$ to generate the $3^5$ cosets we need. Add each vector in $W_2$ to all codewords in $C$. The cosets' leaders are the $W_2$ vectors that generated them. Name cosets based on their leader's non-zero indices. That is, the cosets $S_0$, $S_{(i,b)}$, and $S_{\{(i,b),(j,b')\}}$, where $b$ and $b'$ are 1 or 2, and $i$ and $j$ are the non-zero trits' indices. These cosets also cover the whole vector space $V(11,3)$ because each word is at Hamming distance 0, 1, or 2 of a unique codeword, and all possibilities of a 1 or 2 trit errors have been considered.

Since the Golay code is a linear code, by considering a linear combination between two vectors $l_1 = w_1\cdot G$ and $l_2 = w_2 \cdot G$, we can obtain
\begin{equation*}
    a\cdot l_1 + b \cdot l_2 = (a\cdot w_1 + b\cdot w_2) \cdot G,
\end{equation*}
which itself is a codeword.

In this partitioning, two vectors are classified as being in the same coset if a codeword is obtained by taking their difference. This fact is useful because it means that any two vectors in the same coset must have a Hamming distance of at least 5. This is because the codewords used to create them met this requirement, and they were both added to the same vector - which happens to be the leader vector of the coset in $W_2$.

Our graph-building strategy involves assigning each coset a vertex. We connect two vertices if their corresponding subsets ($S_x$ and $S_y$) have two vectors ($v_x \in S_x$ and $v_y \in S_y$) with a Hamming distance of 1 ($d_H(v_x,v_y) = 1$).

To prove that this generates a $(243,14,1,2)$ strongly regular graph, we take a few steps to show this graph satisfies the $k$, $\lambda$, and $\mu$ parameters' conditions.

First, without loss of generality, consider the all-zero vector $z = (0, 0, \cdots, 0)$. Vector $z$ connects $S_0$ to cosets containing vectors
\begin{align*}
    N_z = \{&(x,0,0, \cdots, 0, 0),\\
    &(0,x,0, \cdots,0, 0), \\
    \cdots, &(0,0,0, \cdots,0,x)\},
\end{align*}
for all $x \in \text{GF(3)} - {0}$, because these vectors are at a Hamming distance of one from $z$; This shows that the degree $deg(S_0)$ is at least 22. To prove no more neighbours exist, consider a vector $y'$ in coset $S_{\{(i,b),(j,b')\}}$. By reducing the leader from $y'$, we obtain a codeword $y$, which is at least at Hamming distance 5 of codeword $x \in S_0$ where $x \neq y$. Therefore, the Hamming distance $d_H(x,y') \geq 3$, and the two cosets have no vectors of distance one with each other. By generalising this approach for all other $W_2$ vectors, we can see that all cosets (vertices) have 22 neighbours. Hence, the graph is 22-regular.

Next, without loss of generality, consider vector $x_+ = (+x,0,0, \cdots, 0, 0)$, one vector of distance 1 from $z$. We can observe that since both $z$ and $x_+$ vectors are at distance 1 of vector $x_- = (-x,0,0, \cdots, 0, 0)$, they cause cosets $S_0$, $S_{(1,x)}$, and $S_{(1,-x)}$ to form a triangle. All other vectors in set $N_z$ have a Hamming distance of two from $x_+$ and $x_-$. Thus, their respective cosets are not adjacent, and consequently, a complete matching with 11 edges exists in $S_0$'s neighbourhood. As a result, the $\lambda = 1$ condition is satisfied.

Also, for a coset that's not $S_0$'s neighbour, its leader, vector $xy_+$, must have weight 2. Without loss of generality, consider these two non-zero elements at indices 1 and 2, which means $xy_+ = (x,y,\cdots,0) \in W_2$. This vector $xy_+$ is at distance 1 of vectors $x_+$ and $y_+$. Hence, the coset $S_{(1,x),(2,y)}$ is adjacent to $S_{(1,x)}$ and $S_{(2,y)}$. That being the case, the $\mu = 2$ condition is satisfied.

\subsubsection{The parity approach} \label{parityApproach243}
A code's parity check matrix describes the conditions a word must satisfy to be a valid codeword. The parity matrix of the Golay code is
\begin{align*}
    H = 
    \left[ 
    \begin{array}{cccccc|ccccc}
    2 & 2 & 2 & 1 & 1 & 0 & 1 & 0 & 0 & 0 & 0 \\
    2 & 2 & 1 & 2 & 0 & 1 & 0 & 1 & 0 & 0 & 0 \\
    2 & 1 & 2 & 0 & 2 & 1 & 0 & 0 & 1 & 0 & 0 \\
    2 & 1 & 0 & 2 & 1 & 2 & 0 & 0 & 0 & 1 & 0 \\
    2 & 0 & 1 & 1 & 2 & 2 & 0 & 0 & 0 & 0 & 1 \\
    \end{array}
    \right]. 
\end{align*} For a word $w$, to be a valid codeword, it needs to satisfy the equation $H\cdot w = (0, 0, 0, 0, 0)$.

Each column of matrix $H$ gives us a ternary vector of length 5 in vector space $V(5,3)$. Naming these vectors $x_1$, $x_2$, $\cdots$, $x_{11}$, we can obtain 11 more vectors by the negation of these vectors. Also, 220 more vectors can be obtained by $\pm x_i \pm x_j$ for all $i,j \in \{1,2,\cdots,11\}$ such that $i < j$. By proving that these 242 vectors are distinct and by adding $z$ to this set (name the set $\theta$), we can obtain the whole set of vectors available in vector space $V(5,3)$.

Since every codeword (but word $z$) has at least weight 5 (due to the distance two condition), it can be inferred every 4 columns of the $H$ matrix must be independent. To prove this assume that there exist four vectors $x_i, x_j, x_k, x_l$ (where \(i, j, k, l \in \{1, \cdots 11\}\) are distinct) such that \(a x_i + b x_j + c x_k + d x_l = 0\) for some \(a, b, c, d \in GF(3)\), but not all zero. Based on this, we can build a word $w$, which is $a, b, c, d$ at indices $i, j, k, l$, respectively. Now, the equation $H \cdot w = 0$  is satisfied, and $w$ becomes a codeword. However, we know each codeword has a weight of 5, leading to a contradiction. As a result, no two vectors in $\theta$ are the same, and $\theta = V(5,3)$.

Now consider vectors $V(5,3)$ as the vertices of a graph. Note that every two neighbours on this graph will have corresponding vectors that differ by a vector $\pm x_1, \pm x_2, \cdots, \pm x_{11}$.

We can simply infer each vertex has 22 neighbours, as we showed all these summations lead to different values. For any $v$, and vector $x_i$, three vertices $v$, $v + x_i$, and $v - x_i$ form a triangle (note that $v - x_i - x_i = v + x_i$ in $GF(3)$). These 7 edges that form a perfect matching are the only available edges in $N(v)$, and thus, the $\lambda = 1$ condition holds. By adding or subtracting two vectors from vertex $v$, $v \pm x_i \pm x_j$, we can reach any other vertex of the graph, and also the vertex $v \pm x_i \pm x_j$ share two neighbours with $v$, which are $v \pm x_i$ and $v \pm x_j$. Hence, the $\mu = 2$ condition is satisfied as well.

We have seen how beautifully a strongly regular graph was derived from a perfect code. More ways of building this graph can be found in \cite{berlekamp1973strongly}.

\section{Subgraphs and patterns}
This section aims to explore the possibilities of different subgraphs and patterns, in a possible Conway-99 graph. The goal is to reduce the search space we have to explore.

\subsection{Paley(9) pattern}
There is a special property associated with the Paley(9) graph. This can be the only possible subgraph that satisfies the $\lambda = 1$ and $\mu = 2$ conditions within the subgraph itself. This is because there are no other strongly regular graphs with $\lambda = 1$ and $\mu = 2$ with less than 99 vertices.

\begin{figure}
    \centering
    \includegraphics[width=0.43\textwidth]{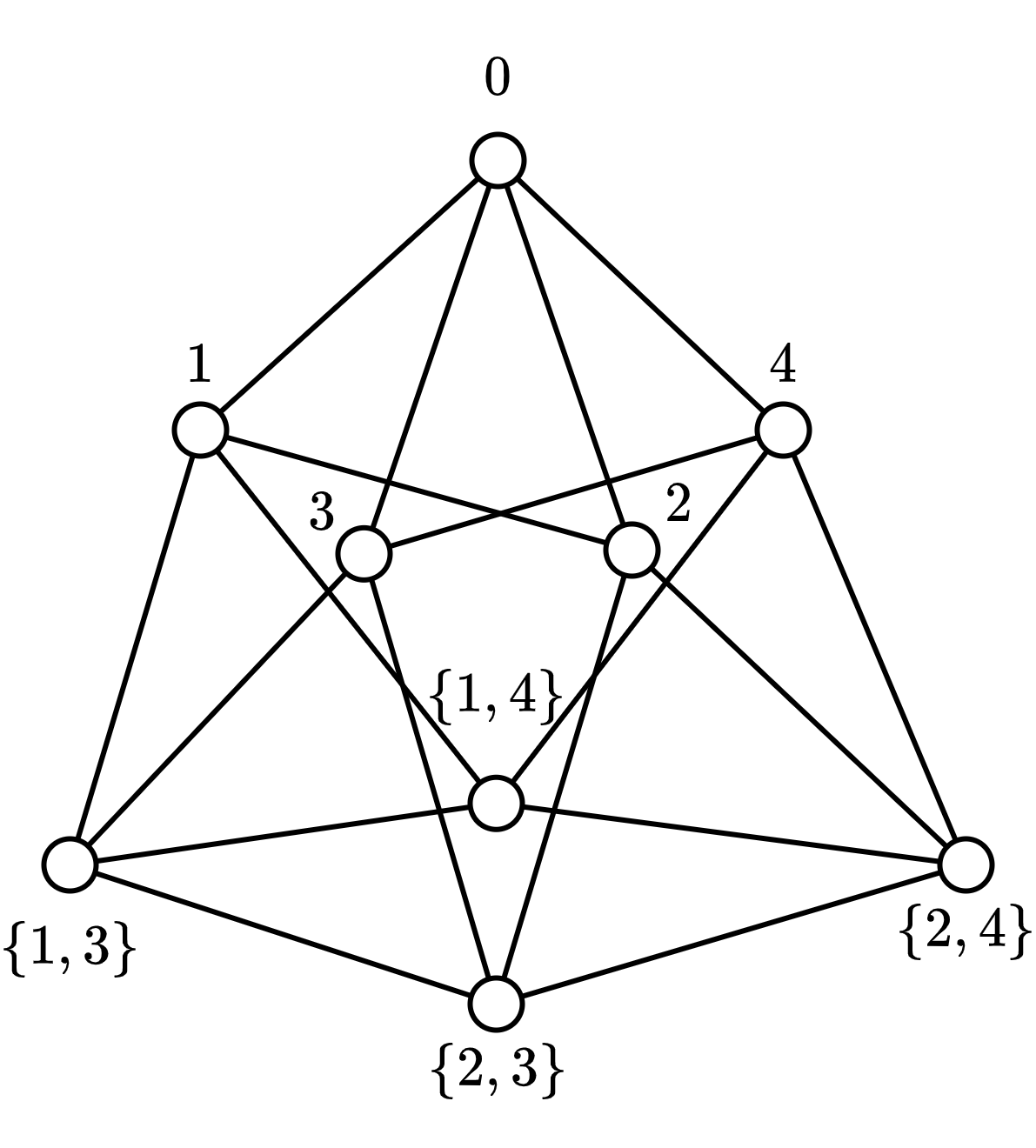}
    \caption{Paley(9) graph: A representation of the graph emphasising the quadrilateral formations in the structure of the graph.}
    \label{paley9quad}
\end{figure}

To make the writing simpler, in strongly regular graphs with $\lambda = 1$ and $\mu = 2$, for a vertex $v$, we sometimes refer to $v$'s non-neighbours as a pair of its neighbours since each non-neighbour is adjacent to two specific neighbours. For example, in figure \ref{paley9quad}, vertex $\{1,3\}_0$ is $0$'s non-neighbour which is adjacent to $1$ and $3$ in $N(0)$.
\begin{definition}[Paley(9) pattern]
The Paley(9) pattern is established in an $(n,k,1,2)$ strongly regular graph $G$ when for all vertices $v\in V$ and two edges $\{v_1,v_2\},\{v_3,v_4\} \in E(N(v))$, the induced subgraph by vertices
\begin{equation*}
    P_v = \{v,v_1,v_2,v_3,v_4,(v_1,v_3)_v,(v_1,v_4)_v,(v_2,v_3)_v,(v_2,v_4)_v\}
\end{equation*} would be a Paley(9) graph.
\end{definition}

\begin{lemma}
    The Paley(9) pattern is present in the Berlekamp-Van Lint-Seidel graph.
\end{lemma}
\begin{proof}
    In the parity approach \ref{parityApproach243}, we noted that a vector $v$ has 22 neighbouring vectors, which are $v \pm x_i$ for $1 \leq i \leq 11$. By considering a vector $v$ and two parity matrix columns, $x_i$ and $x_j$, it is clear that the subgraph shown in \ref{paley9parity} will be created.
    \begin{figure}
    \centering
    \includegraphics[width=0.7\textwidth]{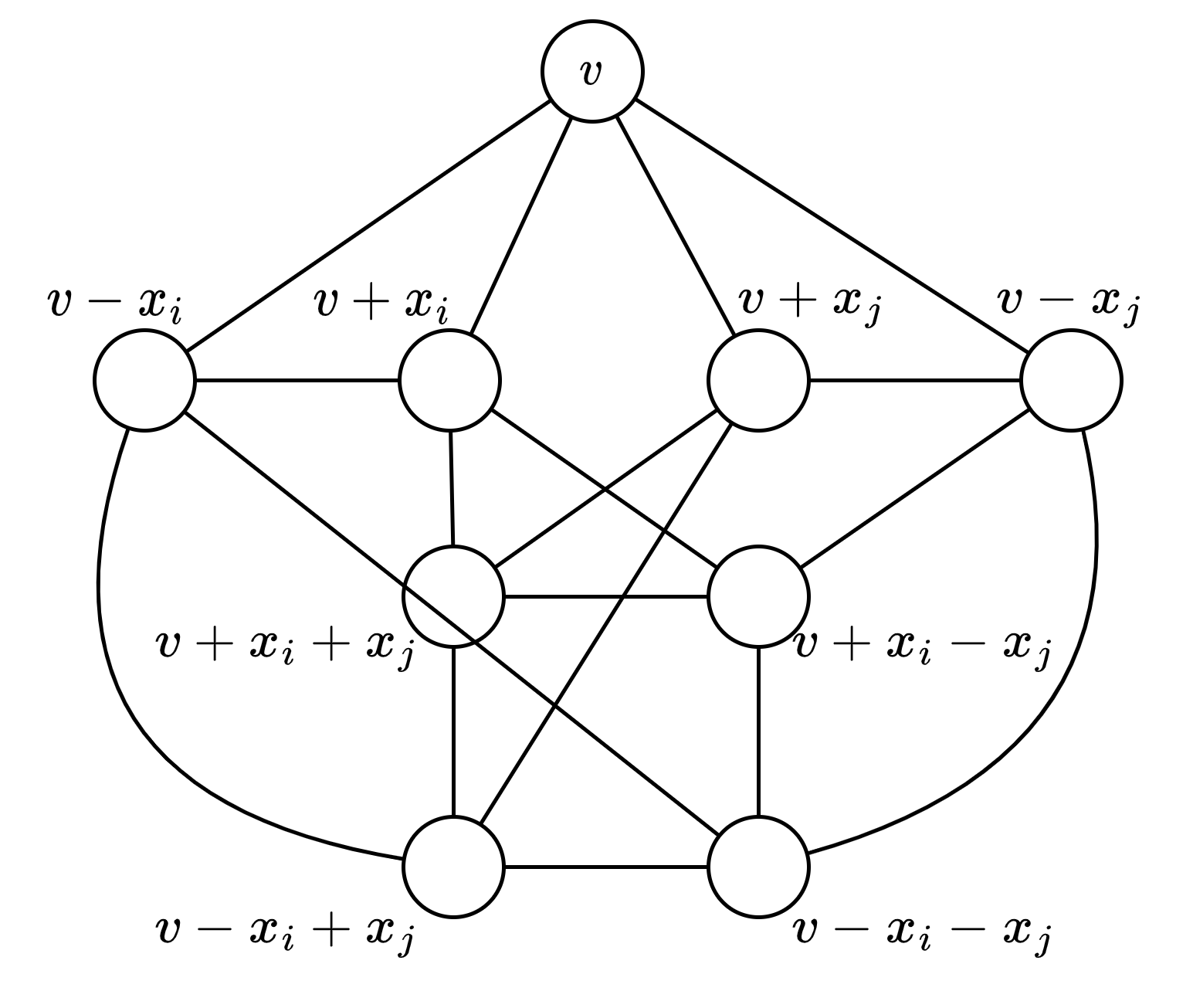}
    \caption{Paley(9) graph: A representation obtained from the Golay ternary code parity matrix.}
    \label{paley9parity}
    \end{figure}
\end{proof}

We now prove this pattern cannot appear in any possible $(99,14,1,2)$ strongly regular graphs.

\begin{theorem} \label{p9patternImpossible}
    If a $(99,14,1,2)$ strongly regular graph exists, it cannot follow the Paley(9) pattern.
\end{theorem}

\begin{figure}[h]
        \centering
        \includegraphics[width=0.7\textwidth]{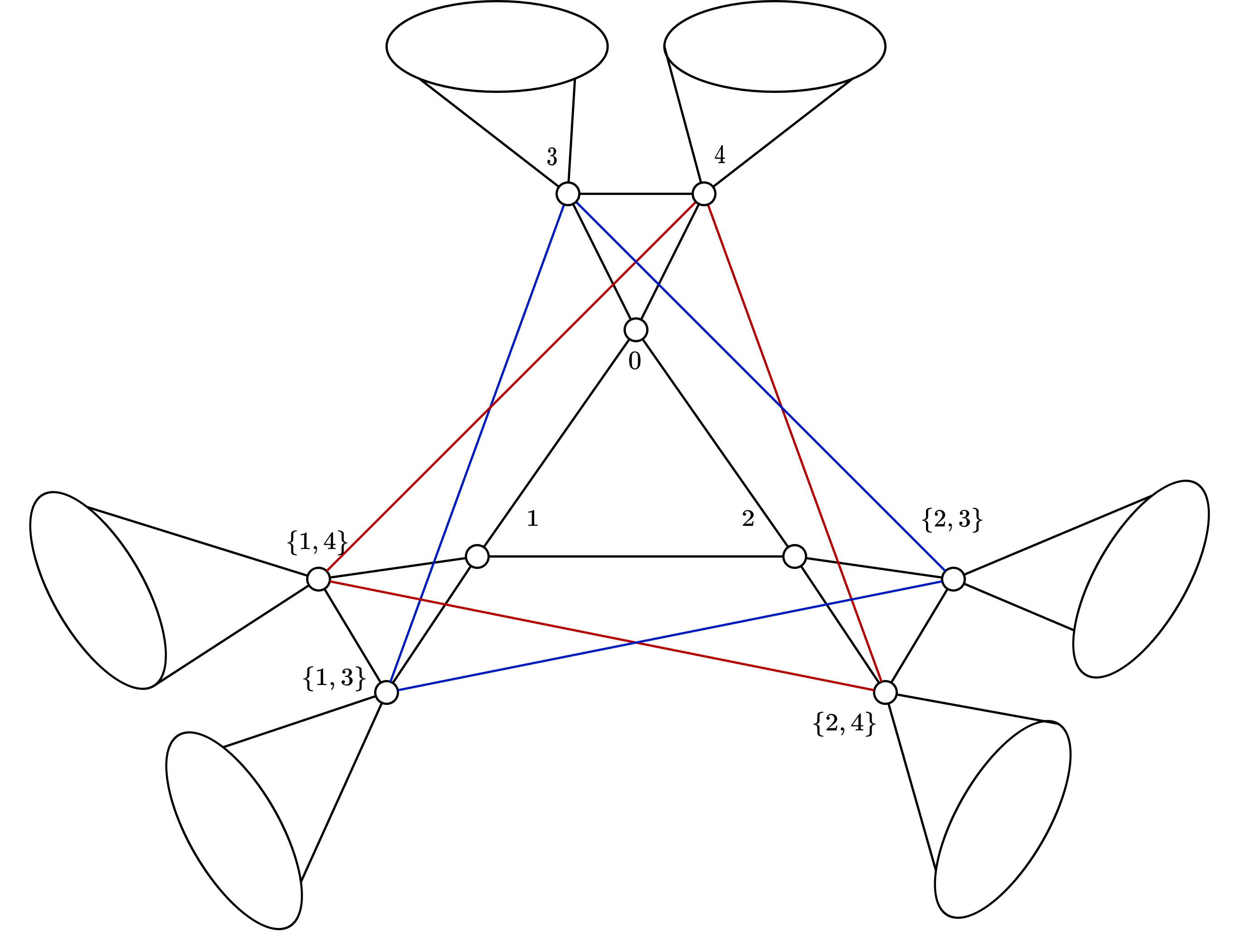}
        \caption{Building a $(99,14,1,2)$ strongly regular graph based on the Paley(9) pattern.}
        \label{srg_9_pattern}
\end{figure}

\begin{proof}
    We proceed by assuming the pattern holds and then derive a contradiction out of it.  We start off by considering a Paley(9) subgraph, as in Figure \ref{srg_9_pattern}. Name vertex \red{0}'s neighbours with numbers \red{1} to \red{14}, such that vertices $2i - 1$ for $1 \leq i \leq 7$ are adjacent to vertices $2i$.
    
    Consider vertex \red{5} in the graph. According to the $\mu = 2$ requirement, each of the sets $N(i)$ for $1 \leq i \leq 4$ must contain one of \red{5}'s neighbours, named (1,5), $\cdots$, (4,5). Considering two triangles $\{0,1,2\}$ and $\{0,5,6\}$, based on the Paley(9) pattern, we can infer edges
    \begin{equation*}
        \{\{(1,5),(1,6)\}, \{(2,5),(2,6)\}, \{(1,5),(2,5)\}, \{(1,6),(2,6)\}\}
    \end{equation*}
    exist, which form a $C_4$. By applying the same reasoning to triangles $\{0,3,4\}$ and $\{0,5,6\}$ we can obtain more edges. We call this the Paley(9) parallelism approach, as the two triangles $\{0,1,2\}$ and $\{5,(1,5),(2,5)\}$ are adjacent with each other.
    
    Next, we name the vertex adjacent to $(1,5)$ in the sets $N_{(1,3)}$ and $N_{(1,4)}$ as $(1,3,5)$ and $(1,4,5)$. We can now show these two vertices must be adjacent. To do this, consider two triangles $\{1,(1,3),(1,4)\}$ and $\{1,(1,5),(1,6)\}$. We can infer that the edges
    \begin{equation*}
        \{\{(1,3,5),(1,4,5)\}, \{(1,3,5),(1,3,6)\}, \{(1,4,5),(1,4,6)\}, \{(1,3,6),(1,4,6)\}\}
    \end{equation*} exist.

    If we apply the parallelism approach used for vertex \red{0} to other vertices and generalise it, we can deduce that the edge present among the sets $N_{(1,3)}$ and $N_{(2,3)}$ must create a triangle with vertices in $N_3$.
    
    Vertex \red{5} needs to have more neighbours and, therefore, more triangles. Thus far, we have inferred the existence of three triangles in which vertex \red{5} appears, which are
    \begin{equation*}
        \{0,5,6\}, \{5,(1,5),(2,5)\}, \{5,(3,5),(4,5)\}.
    \end{equation*}

    Vertex \red{5} must have two neighbours in $N_{1,3}$. Name one vertex as $(1,3,x)$. This neighbour must form a triangle with \red{5}. Let us consider the possibilities for the third vertex $v_3$ of this triangle:
    \begin{enumerate}
        \item Initially, it is not possible for vertex $v_3$ to be a part of $N_{1,3}$ because it would break the condition $\lambda = 1$, as the vertices \red{5} and $(1,3)$ would become common neighbors of $v_3$ and $(1,3,x)$.
        
        \item The edges connecting $N_{(1,3)}$ and $N_{(1,4)}$ create a triangle with vertices in $N_1$ (parallelism). Thus, the third vertex cannot belong to the set $N_{1,4}$.
        
        \item The edges among $N_{(1,3)}$ and $N_{(2,3)}$ form a triangle with vertex $3$ (parallelism), and cannot be form a triangle with vertex \red{5}.
    \end{enumerate}

    Therefore, the only possibility is to have a triangle between vertices $\{5, (1,3,x),(2,4,y)\}$, where $(2,4,y) \in N_{(2,4)}$.
    
    Consider two triangles $\{5, (1,3,x),(2,4,y)\}$ and $\{5, 0, 6\}$. Without the loss of generality, suppose $(1,3,x)$ is adjacent to vertex $7$. This can simply be generalised for all other vertices, but \red{6}, which will break the $\lambda = 1$ condition. These two triangles must also form a Paley(9) graph as shown in figure \ref{paley9triangular}, such that vertex \red{6} appears in triangle $\{6,u,v\}$. We now examine the possible sets to which vertices $v$ and $u$ can belong.
    \begin{enumerate}
        \item As a result of the $\lambda = 1$ condition, $v \notin N_{2,4}$ and $u \notin N_{1,3}$.
        
        \item As a result of the $\lambda = 1$ condition, $v,u \notin N_0$.
        
        \item As a result of the $\mu = 2$ condition, $v,u \notin N_1,N_2,N_3,N_4$.

        \item As a result of parallelism, $v \notin N_{2,3} \cup N_{1,4}$ and $u \notin N_{1,4} \cup N_{2,3}$

        \item As a result of the $\lambda = 1$ condition, $u \notin N_{1,3}$ and $v \notin N_{2,4}$.
    \end{enumerate}
    \begin{figure}
        \centering
        \includegraphics[width=0.7\textwidth]{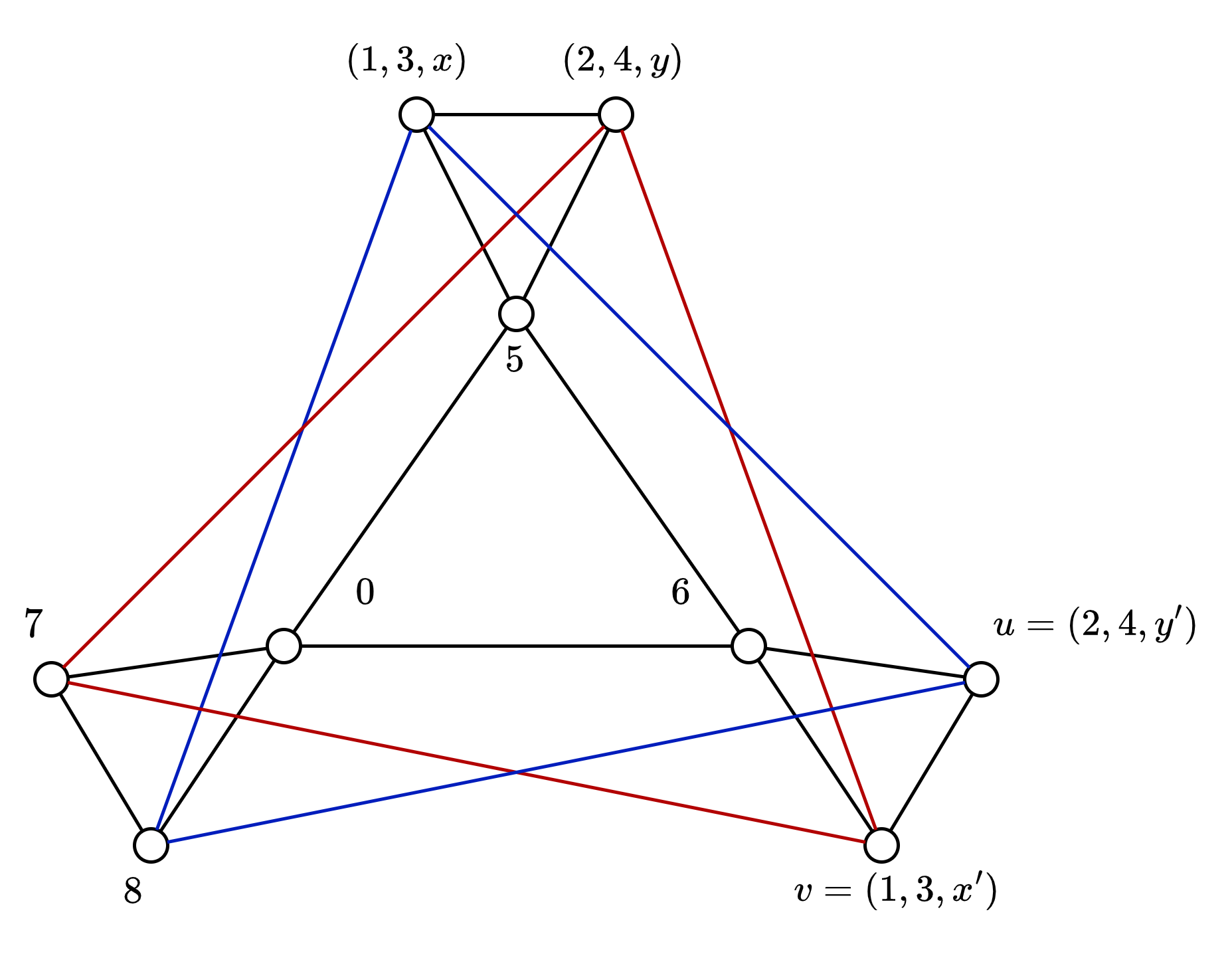}
        \caption{Examining potential neighbouring sets for $v$ and $u$ in a Paley(9) subgraph. This figure also depicts a triangular view of the Paley(9) graph.}
        \label{paley9triangular}
    \end{figure}

    Therefore, the only possibility that remains is that $u$ must belong to $N_{2,4}$ and $v$ must belong to $N_{1,3}$. This leads to a contradiction as $v = (1,3,x')$ and $(1,3,x)$ share three neighbours:
    \begin{equation*}
        \{(1,3), (2,4,y), (2,4,y')\}.
    \end{equation*}
    We can now infer that the Paley(9) pattern cannot hold in a $(99,14,1,2)$ strongly regular graph.
\end{proof}

\subsection{Paley(9) subgraph}

In Theorem \ref{p9patternImpossible}, we proved possible $(99,14,1,2)$ strongly regular graphs cannot follow the Paley(9) pattern. With a similar approach, we can prove a $(99,14,1,2)$ strongly regular graph cannot contain eleven independent Paley(9) subgraphs. We omit the proof here:
\begin{theorem} \label{noElevenPaley9}
    If a $(99,14,1,2)$ strongly regular graph exists, it cannot contain eleven independent Paley(9) subgraphs.
\end{theorem}

By observing these, we conjecture that no Paley(9) subgraphs can exist in a possible $(99,14,1,2)$ strongly regular graph:
\begin{conjecture} \label{noPaley9Subgraph}
    In a $(99,14,1,2)$ strongly regular graph instance, there cannot be any Paley(9) subgraph.
\end{conjecture}
By demonstrating this, we establish the infeasibility of a subgraph that meets the neighbouring conditions \textit{within itself}.

\subsection{Triangular view} \label{triangular_view}
An interesting and possible way to approach the Conway-99 problem is to focus on its triangles. As we studied, the $\lambda = 1$ property makes triangles one of the fundamental units of this graph. In this part, we want to study how the triangles interact with each other and affect the overall graph's structure.

As can be seen in Figure \ref{paley9triangular}, the Paley(9) graph consists of six triangles. To construct the triangular graph, we will follow a two-step process. Firstly, we will assign a vertex to each triangle, and secondly, we connect the triangles that share a vertex.

\begin{lemma}
    The triangular graph of an $(n,k,1,2)$ strongly regular graph, consists of $\frac{nk}{6}$ vertices and is $\frac{3k - 6}{2}$-regular. Every two neighbours share $\frac{k}{2} - 2$ neighbours. Also, every two non-neighbour can share at most three neighbours.
\end{lemma}
\begin{proof}
    In the initial graph with $\lambda = 1$ and $\mu = 2$ parameters, each vertex $v$, has a neighbourhood of $\frac{k}{2}$ edge perfect matching, which means it is inside $\frac{k}{2}$ triangles. By adding up the number of triangles adjacent to each vertex, we get $n\frac{k}{2}$. Since each triangle is counted three times, the actual number of triangles is $\frac{nk}{6}$.
    
    Consider one triangle. Each vertex of it has $\frac{k}{2} - 1$ more triangles adjacent to it. Also, because of the $\mu = 2$ condition, no triangle is counted twice. Therefore, each triangle is adjacent to exactly $\frac{3k - 6}{2}$ triangles.

    Consider two adjacent triangles, that is, they share a vertex $v$. It becomes straightforward to infer that no edge can exist among the four other vertices as such a scenario would violate the $\lambda = 1$ condition. Therefore, the two adjacent triangles are only adjacent to the other $\frac{k}{2} - 2$ triangles of vertex $v$.
    
    % HEREEEEEEE
    Among two non-adjacent triangles of the initial graph, each vertex is adjacent to at most one vertex of the other triangle. Therefore, a maximum of three edges can exist between two non-adjacent triangles.  Hence, in the triangular graph, every non-neighbour pair has at most three shared neighbours.
\end{proof}
Although this is not the original description of the $\mu$ parameters, we use it somewhat differently here. In this graph, we have $\lambda = \frac{k}{2} - 2$, and $\mu \leq 3$. The notation $\mu \leq 3$ here shows that two non-neighbours have at most three common neighbours.

Therefore, for a $(99,14,1,2)$ strongly regular graph to exist, we first must have a graph $G$ with parameter set $(231,18,5, \leq 3)$. 

\begin{observation}
    The triangular graph of Paley(9) is $K_{3,3}$. This graph has 6 vertices, is 3-regular, and is triangle-free.
\end{observation}
This graph is a trivial, strongly regular graph.

\chapter{Computer search using SAT solvers} \label{SATSearch}
Computer search has been historically used to find instances of strongly regular graphs or prove they do not exist. Through the usage of backtracking algorithms, Coolsaet et al. \cite{coolsaet2006strongly} successfully discovered all potential non-isomorphic strongly regular graphs with a parameter set of $(45,12,3,3)$. Behbahani and Lam \cite{BEHBAHANI2011132} used matrices and strongly regular graphs' algebraic structures to study automorphism groups. They stated that they had used multiple years of computation using hundreds of Intel-based Linux machines until they were able to produce their results.

In 1989, Bussemaker et al. \cite{bussemaker198949} used computation power to prove the non-existence of a strongly regular graph with parameters $(49,16,3,6)$. They started by eliminating various possibilities through theoretical analysis and then used a computer to search the remaining space. Searching through that space, although not computationally intensive, was not doable by hand.

A fairly recent result was obtained by Gritsenko \cite{gritsenko2021strongly}, where a $(65, 32, 15, 16)$ strongly regular graph was discovered. The idea was to construct a matrix with blocks and then circulate through those blocks. More methods and results can be seen in \cite{spence2000} \cite{mckay2001classification}. In \cite{srgDatadabase}, Brouwer has compiled a great amount of data on strongly regular graphs.

In this chapter, we will delve into the SAT problem and explore how to encode the search problem for strongly regular graphs as a SAT formula. We will cover multiple types of SAT formulas for different encodings.

\section{SAT problem and SAT solvers} \label{satSolver}

SAT solvers are considered black-box tools in this research project. Although sufficient in many cases, it does not offer the level of information and efficiency needed to solve the Conway-99 problem.

The following is a basic definition of the 
\textit{Boolean satisfiability}, abbreviated as \textit{SAT}, problem:
\begin{definition}[SAT problem]
    Given a Boolean formula $F$ with variables $x_1, x_2, \dots, x_n$ and logical operators $\land$ (AND), $\lor$ (OR), and $\neg$ (NOT), is there an assignment of truth values ${0,1}$ to the variables such that $F$ evaluates to \syntactic{True}? If so, the formula $F$ is called satisfiable, and unsatisfiable otherwise.
\end{definition}

Boolean formulas can be represented in many forms; however, the most commonly used representation is the {\it conjunctive normal form} (CNF). 

\begin{definition}[CNF]
    A CNF formula is a conjunction (logical AND) of multiple clauses, such that each clause is a disjunction (logical OR) of literals. Literals are either variables or their negation.
\end{definition}

The use of CNF in SAT is simple and straightforward, providing a clear representation of the formula. In fact, that is why many SAT solvers take their inputs in CNF; It should also be noted that any SAT formula can be transformed into CNF using boolean algebra.

One simple formula can be the following:
\begin{equation*}
        \phi = (x_1 \land x_2) \lor \neg x_3,
\end{equation*}
which can be transformed into a CNF version:
\begin{equation*}
        \phi = (x_1 \lor \neg x_3) \land (x_2 \lor \neg x_3).
\end{equation*}

SAT solvers use different algorithms to determine the satisfiability of a formula.

In complexity theory, the SAT problem plays a pivotal role. Since NP-Hard problems are all reducible to a SAT instance in polynomial time, SAT solvers have numerous applications. The development and advancements in SAT solver algorithms have led to more efficient and faster solutions over the years. SAT solvers are a powerful tool widely used in various domains, such as hardware and software verification, planning and scheduling, cryptography, and artificial intelligence. In recent years, the development of parallel and distributed SAT solvers has enabled the solving of even larger and more complex instances of the SAT problem. A specific competition is held each year, known as the international SAT competition, to test, benchmark, compare, and suggest new SAT algorithms.

The development of the Davis-Putnam-Logemann-Loveland (DPLL) algorithm, which is a backtracking algorithm, was a key moment in the history of SAT solvers. Since then, many variations and heuristics have been applied to explore the space of possible assignments more efficiently. These improvements include conflict-driven clause learning (CDCL), non-chronological backtracking, and clause deletion.

One of the most significant improvements to SAT algorithms has been the use of preprocessing techniques to simplify the input formula before solving it \cite{biere2021preprocessing}. Preprocessing can involve various operations such as unit propagation (removing clauses that contain single literals), clause learning, variable elimination, and symmetry breaking. These techniques can reduce the size of the input formula and simplify the search space, leading to faster solving times and more efficient use of resources.

Conflict-Driven Clause Learning (CDCL) SAT solvers work by incrementally assigning values to boolean variables. These solvers maintain the constraints with another set of clauses derived from the previous set. When a conflict arises, meaning the algorithm can no longer continue assigning variables to satisfy the SAT formula, the algorithm backtracks and learns from the conflict. In the learning process, it adds a new clause that blocks the previous assignment causing the conflict. This continues until a solution is found or it is determined that no valid assignment exists. Algorithm \ref{CDCLPSEUDO} presents a simple pseudo-code of the process.

\begin{algorithm}[H]
  \SetKwFunction{CDCL}{CDCL}
  \SetKwProg{Fn}{Function}{:}{}
  \SetKwInOut{Input}{Input}
  \SetKwInOut{Output}{Output}

  \Fn{\CDCL{Formula \syntactic{F}}}{
    \While{\syntactic{F} variables are not entirely assigned}{
        \syntactic{state} $\gets$ save the state\\
        Select a variable in \syntactic{F} and assign it to \syntactic{True} or \syntactic{False}\\
        Simply do a unit propagation\\
        \syntactic{conflict} $\gets$ Find the conflicts in \syntactic{F}\\
        \If{\syntactic{conflict} = $\emptyset$}{
            continue \\
        }
        \syntactic{clause} $\gets$ derive a new clause from the conflict\\
        Add \syntactic{clause} to \syntactic{F} \\
        backtrack to \syntactic{state}
    }
  }
  
  \caption{The Conflict-Driven Clause Learning (CDCL) Algorithm}
  \label{CDCLPSEUDO}
\end{algorithm}

One great example of a practical SAT solver is CaDiCaL \cite{BiereFazekasFleuryHeisinger}, an open-source project also available on GitHub. It is a high-performance solver used in many applications.

The huge search space involving the Conway-99 problem is deemed impractical to search at the moment. With basic exhaustive approaches, there are a total of $2^{\binom{99}{2}}$ potential graphs; obtained by placing or not placing each edge out of the $\binom{99}{2}$ edges.

In this research, we have implemented tools and programs for encoding the Conway-99 problem and, more generally, encoding the search for all strongly regular graphs as SAT problems. Two approaches are explained in Sections \ref{cnfSection} and \ref{PBSection}. However, before getting into the problem's encoding, we explain in Section \ref{symmetrySection} why we are utilising SAT solvers and what makes them special.

\section{Symmetry-breaking techniques} \label{symmetrySection}

In simple words, symmetry-breaking tries to avoid considering multiple identical copies of the same object. It avoids redundant computation by ensuring the solver does not explore multiple solutions that are essentially the same due to their symmetry.

This approach is usually made by adding symmetry-breaking clauses to our SAT instance. The following example illustrates this and similar approaches:
\begin{example}
    Consider a SAT problem that involves three boolean variables, namely $x$, $y$, and $z$:
    \begin{equation*}
        (x \lor y) \land (x \lor \neg z) \land (y \lor \neg z) \land (\neg x \lor \neg y \lor z).
    \end{equation*}
    The symmetry between $x$ and $y$ is visible. By swapping these, the SAT problem remains unchanged. One way to break this symmetry is to add the clause $(\neg x \lor y)$. By considering this clause, the case $x = \syntactic{True}$ and $y = \syntactic{False}$ is not considered.
\end{example}

However, symmetry detection itself is not an easy task. This problem is essentially equivalent to the graph isomorphism problem. The graph isomorphism problem asks whether the two input graphs are isomorphic. This problem is classified as NP \cite{luks1996symmetry}.

In the realm of graph theory, SAT solvers have been employed to solve many problems. One common usage is in graph colouring problems, which contain many symmetrical answers to it.

When dealing with strongly regular graphs, the application of SAT solvers can prove to be valuable in breaking symmetries. Many of these graphs possess a great amount of symmetry, which can be eliminated through the addition of extra conditions. This process can effectively decrease the search area.

\begin{figure}
  \centering
  \includegraphics[width=\linewidth]{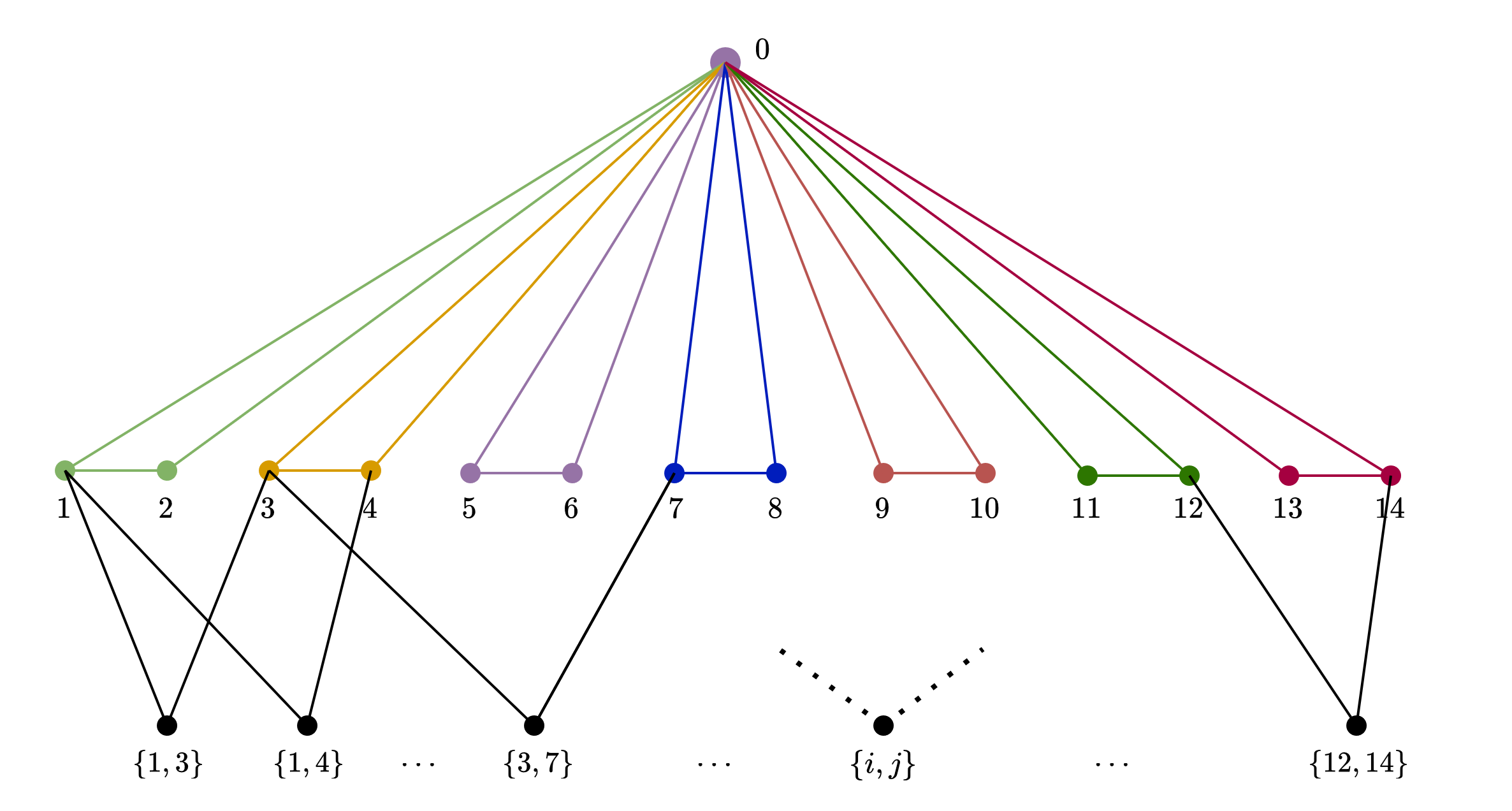}
  \caption{Conway99 graph's view from one vertex. Each non-neighbour can be represented as a set with cardinality two.}
  \label{oneVertex}
\end{figure}

A basic symmetry-breaking example for the Conway-99 problem can be the graph's view from one vertex, as shown in Figure~\ref{oneVertex}. Consider a possible solution as graph $G = (V,E)$. For a vertex $v \in V$, and two neighbours of it $i,j \in N(v)$ where $i$ and $j$ are not connected, there must exist a unique vertex in $V - N[v]$ connected to $i$ and $j$, so that the $\mu = 2$ condition between vertices $i$ and $j$ is satisfied. Also, again because of the $\mu = 2$ condition, this vertex cannot be connected to any other vertices in $N(v)$. Label this vertex as $\{i,j\}$. We can now label all vertices in $V - N[v]$ with an unordered pair of non-neighbour vertices in $N(v)$. By labelling these vertices as shown in Figure~\ref{oneVertex}, the problem reduces to finding the edges among the $\frac{14}{2} - 7 = 84$ unordered pairs.

\section{CNF Clauses} \label{cnfSection}
In this section, we encode our problem into a CNF instance.

\begin{lemma} \label{cnfMinimum}
    By employing $\binom{n}{\theta + 1}$ CNF clauses, we can generate a formula that is \syntactic{True} when at most $\theta$ \syntactic{True} instances exist among $x_1, x_2, \ldots, x_n$ boolean variables.
\end{lemma}
\begin{proof}
    For each subset of size $\theta + 1$ in variables, we add the following clause (without the loss of generality, consider the subset as $\{x_1, \ldots, x_{\theta + 1}\}$):
    \begin{equation*}
        (\neg x_1 \lor \neg x_2 \lor \ldots \lor \neg x_{\theta + 1}),
    \end{equation*} to the final formula $F$. The formula validates to \syntactic{True} if we have $\theta$ or less \syntactic{True} variables.
\end{proof}

\begin{lemma} \label{exactlyCNFNumber}
    By employing $\binom{n}{\theta - 1} + \binom{n}{\theta + 1}$ CNF clauses, we can generate a formula that is \syntactic{True} when exactly $\theta$ \syntactic{True} instances exist among $x_1, x_2, \ldots, x_n$ boolean variables.
\end{lemma}
\begin{proof}
    Similar to the Lemma \ref{cnfMinimum}, we will use $\binom{n}{\theta + 1}$ to ensure we have at most $\theta$ \syntactic{True} variables.

    To proceed, we utilise $\binom{n}{n - \theta + 1}$ more clauses to guarantee that only a maximum of $n - \theta$ variables are marked as \syntactic{False}. The CNF formula generated employs $\binom{n}{\theta - 1} + \binom{n}{\theta + 1}$, ensuring that precisely $\theta$ variables are identified as \syntactic{True}.
\end{proof}

We start by encoding the small, simple Paley(9) graph as a starting point for the encoding schemes of strongly regular graphs.

\subsection{Paley(9)}

In order to begin the encoding process, we at least need $\binom{9}{2} = 18$ variables; each variable represents an edge and is denoted by the notation $e_{i,j}$. These boolean variables determine whether an edge exists between vertices $i$ and $j$ (vertices are numbered $1$ to $n$), with a value of \syntactic{True} indicating its existence.

To ensure the $4$-regularity condition holds, for each vertex $v$, we need to ensure that exactly $4$ boolean variables out of the set $\{e_{1,v}, e_{2,v}, \ldots, e_{n,v}\} - \{e_{v,v}\}$ are \syntactic{True}. Using the approach in Lemma \ref{exactlyCNFNumber}, we can use $9 \cdot (\binom{13}{3} + \binom{13}{5})$ clauses to ensure regularity.

To encode shared neighbours, we use Cherry subgraphs (other names include $P_3$, claw, and angle). The boolean variables $c_{i,j,k}$, for $i, j, k \in \{1, 2, 3, \ldots, n\}$ and $i < j$, represent the existence of Cherry in Figure \ref{cherry_angle}.

\begin{figure}
    \centering
    \includegraphics[width=0.25\textwidth]{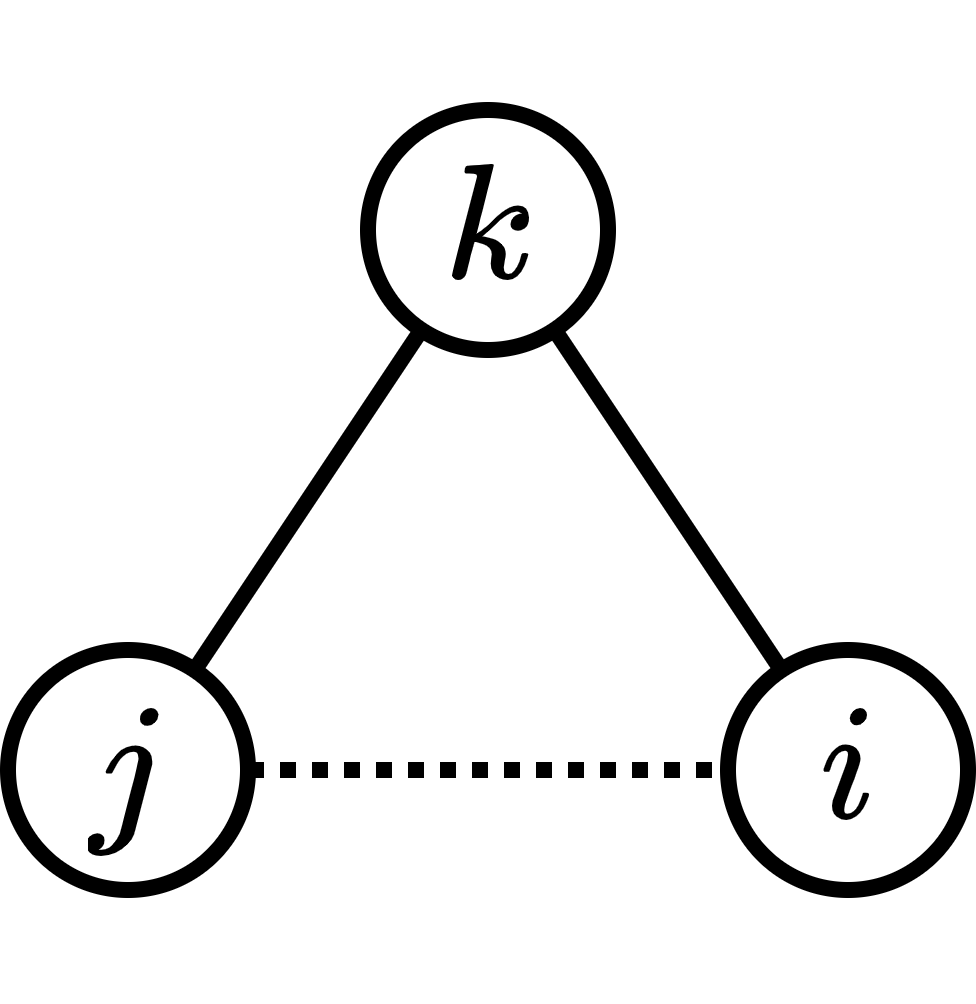}
    \caption{A shared neighbour $k$ for vertices $\{i,j\}$. The edge $\{i,j\}$ may or may not be present. }
    \label{cherry_angle}
\end{figure}

To encode Cherries' existence, we need $\binom{n}{2} (n - 2)$ number of variables. Each Cherry variable $c_{i,j,k}$ holds, if both edges $e_{i,k}$ and $e_{j,k}$ are \syntactic{True}, regardless of the variable $e_{i,j}$:
\begin{equation*}
    (e_{i,k} \land e_{j,k}) \longleftrightarrow c_{i,j,k} \equiv (\neg e_{i,k} \lor \neg e_{j,k} \lor c_{i,j,k}) \land (\neg c_{i,j,k} \lor e_{j,k}) \land (\neg c_{i,j,k} \lor e_{i,k}) .
\end{equation*}
If the edge $e_{i,j}$ holds, we need to satisfy the $\lambda = 1$ condition. To do this, we consider $\binom{n - 2}{\lambda + 1} + \binom{n - 2}{\lambda - 1}$ number of clauses to make sure exactly $\lambda = 1$ Cherries exist among the vertices $i$ and $j$. Then, by adding a $\neg e_{i,j}$ literal to these clauses, we make sure this only applies when the edge appears.

Similar to the $\lambda$ condition, to satisfying the $\mu = 2$ condition involves having $\binom{n - 2}{\mu - 1} + \binom{n - 2}{\mu + 1}$ clauses per edge. To ensure these conditions are met when $e_{i,j} = \syntactic{False}$, we simply include the literal $e_{i,j}$ in each clause.

Hence, in total for each edge, we have made:
\begin{equation*}
    \binom{n - 2}{\mu - 1} + \binom{n - 2}{\mu + 1} + \binom{n - 2}{\lambda + 1} + \binom{n - 2}{\lambda - 1}
\end{equation*}
number of clauses.

\subsection{Conway-99}
Considering this naive approach, the number of clauses needed is huge for this graph and practically cannot be generated. The number of clauses, just to ensure regularity for one vertex, is $\binom{98}{13} + \binom{98}{15}$ which is approximately $\approx 1.8 \times 10^{17}$. However, in Theorem \ref{removeKregular}, we prove that $k$-regularity is an excessive condition; hence, many clauses can be removed.

\begin{theorem} \label{remove14regular}
    Let $G = (V, E)$ denote a graph with 99 vertices. If all pairs of neighbouring vertices share $\lambda = 1$ neighbours, and every two non-neighbours have two common neighbours $\mu = 2$, then graph $G$ is 14-regular. Consequently, $G$ is a strongly regular graph.
\end{theorem}

\begin{proof}
Consider a vertex $v$ with $d_v$ neighbours. Since the induced subgraph of $N(v)$ has to be a complete matching, degree $d_v$ must be even. Also, each pair of non-neighbours in $N(v)$ needs to share a neighbour in $V - N[v]$. These neighbours must be distinct; if not, there would be a vertex $x \in V - N[v]$ that is connected to more than 2 neighbours of $N(v)$. Thus, the following equality must hold:
\begin{equation*}
    |V - N[v]| = \binom{d_v}{2} - \dfrac{d_v}{2} \Rightarrow 99 - d_v - 1 = \binom{d_v}{2} - \dfrac{d_v}{2} \Rightarrow d_v = 14,
\end{equation*}
and all vertices are of degree 14.
\end{proof}

If we limit our consideration to clauses only necessary for the $\lambda$ and $\mu$ parameters, as well as those needed for Cherries, we would still need $\approx 7e8$ number of clauses, which requires 22GB, a huge amount of storage.

To further limit the storage and the number of clauses needed, we decided to stabilise one vertex and its neighbourhood. In figure \ref{oneVertex}, by stabilising one vertex and its neighbours, we end up with exactly 84 vertices, such that each can be represented by an unordered pair $\{i,j\}$, where $i,j \in \{1 \ldots 14\}$ and $\{i,j\} \notin N[v]$. We can use this as a starting point to create a viable input for SAT solvers to handle. Unfortunately, the search space is still huge, and CNF SAT solvers cannot entirely explore it. 

\section{Pseudo-Boolean SAT clauses} \label{PBSection}

Unlike the CNF format, the pseudo-boolean (PB) representation is much more efficient in our problem. The pseudo-boolean form is a more generalised form of CNF, which allows for linear equations, equalities or inequalities. In these problems, the goal is to assign boolean values to the variables to satisfy all the equations. The problems in PB format are similar to constraint satisfaction problems (CSP) but with the condition that the domain of their variables is binary bits.

Another usefulness of this format is its optimisation capabilities, where the goal is, along with satisfying the assignments, to minimise or maximise certain linear objective functions.

When dealing with the Conway-99 problem, the regularity condition used to be a stopping point. However, we can simply overcome this problem with linear equations. To better understand this, let us consider the following example:
\begin{example} \label{linearExample}
    In order to express the $k$-regularity requirement of a vertex $v$, we can represent it as follows:
    \begin{equation*}
        +1\ e_{v,1} +1\ e_{v,2} +1\ e_{v,3} + \ldots +1\ e_{v,n} = k.
    \end{equation*}
    Applying this to all vertices, we can guarantee the $k$-regularity. 
\end{example}
Note that the variable $e_{v,v}$ does not exist due to the absence of self-loops.

In pseudo-boolean encoding, there are two different ways to encode the problem: linear and non-linear. A linear instance is similar to what we previously saw in Example \ref{linearExample}. In non-linear instances, it is possible to have the multiplication of variables in the constraints. In the binary form, this is the same as a conjunction. We now encode our problems in both methods.

\subsubsection{Non-linear encoding}
The minimum number of variables $\binom{n}{2}$ is enough for non-linear instances.

\begin{lemma} \label{nonLinearConditionsLemma}
    The equation
    \begin{equation} \label{nonLinearPBLambdaEquation}
        -\lambda\ e_{i,j} +1\ e_{i,j}e_{i,1} e_{1,j}  +1\ e_{i,j}e_{i,2} e_{2,j} + \ldots
        +1\ e_{i,j}e_{i,n} e_{n,j} = 0
    \end{equation} verifies that the $\lambda$ parameter condition for the pair $\{i,j\}$ holds. Also, the equation
    \begin{multline} \label{nonLinearPBMuEquation}
        +1\ e_{i,1}e_{1,j} +1\ e_{i,2}e_{2,j} + \ldots +1\ e_{i,n}e_{n,j}\\
        -1\ e_{i,j}e_{i,1} e_{1,j} -1\ e_{i,j}e_{i,2} e_{2,j} - \ldots -1\ e_{i,j}e_{i,n} e_{n,j} +\mu e_{i,j} = \mu
    \end{multline} verifies the $\mu$ parameter condition.
\end{lemma}
\begin{proof}
    The multiplication of $e_{i,k}$ and $e_{k,j}$ represents the existence of Cherry formed by vertex $k$ among vertices $\{i,j\}$. 
    
    In case the edge $e_{i,j}$ exists, the equation $\Sigma_k e_{i,k}e_{j,k} = \lambda$ must hold, that is, we need $\lambda$ Cherries. Otherwise, the equation $\Sigma_k e_{i,k}e_{j,k} = \mu$ needs to hold.

    Equation \ref{nonLinearPBLambdaEquation} validates to $0$ when $e_{i,j}$ is \syntactic{False}. However, when $e_{i,j}$ is \syntactic{True}, it ensures the $\Sigma_k e_{i,k}e_{j,k} = \lambda$ equality.

    Equation \ref{nonLinearPBMuEquation} holds whenever $e_{i,j}$ is \texttt{True}. Otherwise, the second part of the equation evaluates to $0$, and $\mu$ angles must exist between $i$ and $j$.
\end{proof}
\subsubsection{Linear encoding}
For a linear pseudo-boolean solver, we define two new sets of variables:
\begin{enumerate}
    \item Angles $a_{i,j,k}$ evaluates to \texttt{True}, when $e_{i,j}$ is \texttt{False}, but the Cherry $e_{i,k}e_{k,j}$ exists.
    \item Triangles $t_{i,j,k}$, which hold \texttt{True} when both the Cherry and the edge exist.
\end{enumerate}

\begin{lemma}
    The following inequalities define the angles we require:
    \begin{align*}
        &-1\ a_{i,j,k} -1\ e_{i,j} \geq -1, \\
        &-1\ a_{i,j,k} +1\ e_{i,k} \geq 0,\\
        &-1\ a_{i,j,k} +1\ e_{k,j} \geq 0, \\
        &+1\ a_{i,j,k} +1\ e_{i,j} -1\ e_{i,k} -1\ e_{k,j} \geq -1,
    \end{align*}
    and for triangles:
    \begin{align*}
        &-1\ t_{i,j,k} +1\ e_{i,j} \geq 0, \\
        &-1\ t_{i,j,k} +1\ e_{i,k} \geq 0, \\
        &-1\ t_{i,j,k} +1\ e_{k,j} \geq 0, \\
        &+1\ t_{i,j,k} -1\ e_{i,j} -1\ e_{i,k} -1\ e_{k,j} \geq -2.
    \end{align*}
\end{lemma}
\begin{proof}
    The first three angles and triangles equations ensure that $a_{i,j,k}$ and $t_{i,j,k}$ can only be \syntactic{True} when the corresponding edges satisfy the needs. The fourth line ensures that the angle and the triangle must be \syntactic{True} when all three edges meet the needs. 
\end{proof}

The rest of the encoding is similar to Lemma \ref{nonLinearConditionsLemma}. For each pair of vertices $ij$, we need the equation
\begin{equation*}
    + \mu e_{ij} +1\ a_{i,j,1} +1\ a_{i,j,2} \ldots +1\ a_{i,j,n} = \mu,
\end{equation*}
when the pair is not adjacent, and the equation
\begin{equation*}
    - \lambda e_{ij} +1\ t_{i,j,1} +1\ t_{i,j,2} \ldots +1\ t_{i,j,n} = 0,
\end{equation*}
when the pair is adjacent.

\subsection{Triangular view}
In Section \ref{triangular_view}, we showed that the existence of a $(99,14,1,2)$ strongly regular graph depends on a graph with parameters $(231,18,5,\mu \leq 3)$. To find this graph, we can encode it using our previous approaches; however, for the $\mu \leq 3$ condition, Equation \ref{nonLinearPBMuEquation} must turn into inequality:
\begin{multline}
    +1\ e_{i,1}e_{1,j} +1\ e_{i,2}e_{2,j} + \ldots +1\ e_{i,n}e_{n,j}\\
    -1\ e_{i,j}e_{i,1} e_{1,j} -1\ e_{i,j}e_{i,2} e_{2,j} - \ldots -1\ e_{i,j}e_{i,n} e_{n,j} +\mu e_{i,j} \leq \mu.
\end{multline}

\subsection{Paley(9) subgraph}
In Theorem \ref{noElevenPaley9}, we showed there cannot be 11 separate Paley(9) subgraphs. To encode this, we simply divided vertices into 11 sets of 9 vertices and set the edge boolean variables to \syntactic{True} for edges and \syntactic{False} for non-edges. However, the SAT solver was unable to infer the unsatisfiability (after a 12-hour running time). This gives us a good example of SAT solvers' weaknesses when it comes down to unsatisfiable formulas, as they will be doing an exhaustive search with minor optimisations. We will go through our experiments more deeply in Section \ref{experiments}.

\begin{corollary}
Since we know the Paley(9) pattern cannot hold, we can adjust our initial configuration and assign edge values to ensure that there is at least one vertex that does not follow the Paley(9) pattern.
\end{corollary}

\section{Parameter redundancy} \label{parameterRedundancy}
In Theorem \ref{remove14regular}, we saw that the $14$-regular condition is redundant; hence, the Conway-99 problem based only on $\lambda$ and $\mu$ parameters is equivalent to the main problem. A question that arises is that can we generalise this?

The Petersen graph and its parameters $(10,3,0,1)$ serve us as a counterexample.
\begin{example}
    The graph depicted in \ref{star10Graph} satisfies both $\lambda = 0$ and $\mu = 1$ conditions. However, it is not a regular graph.
\end{example}

\begin{figure}
    \centering
    \includegraphics[width=0.4\textwidth]{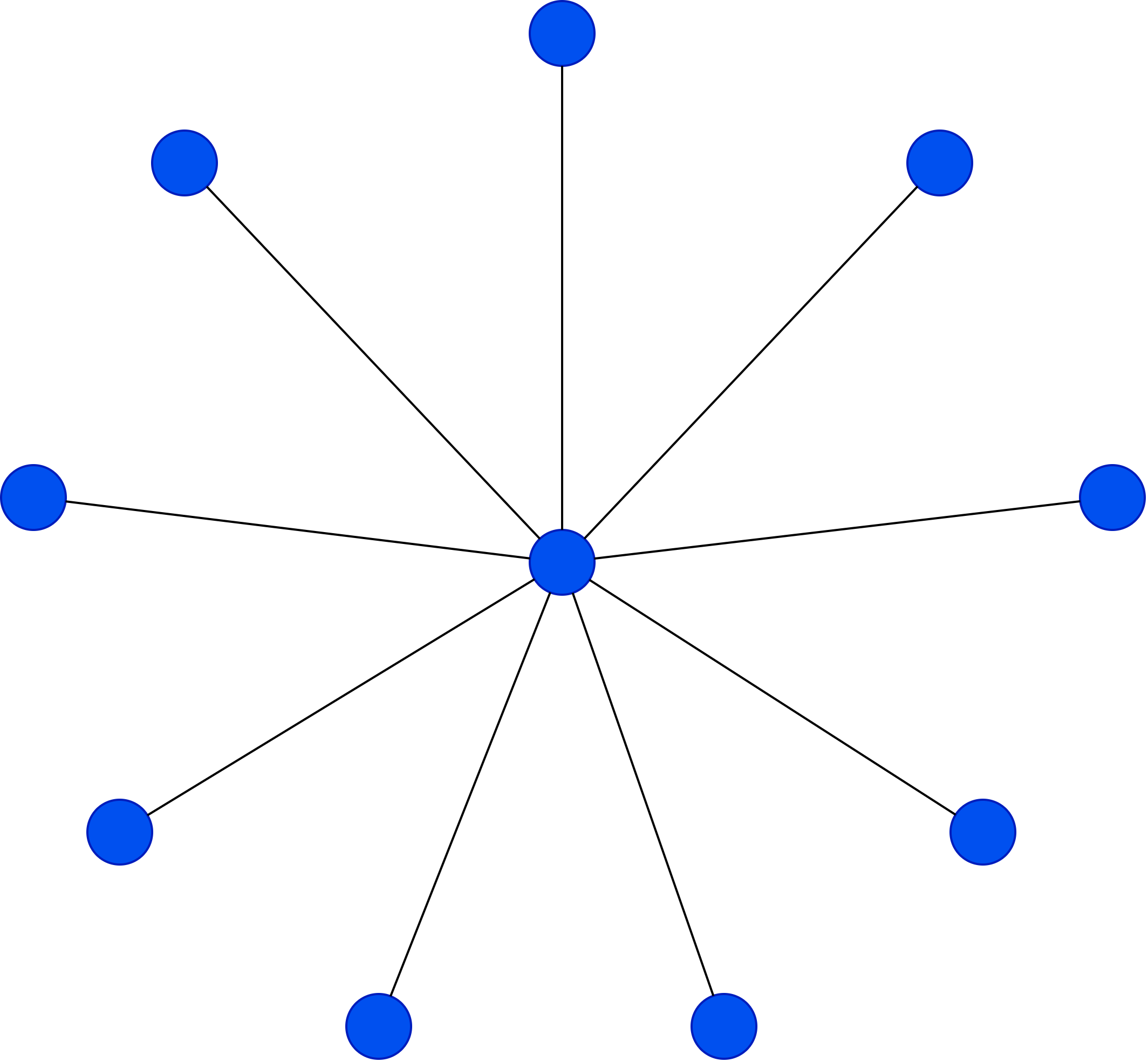}
    \caption{Star graph $K_{1,9}$, a graph with the same $\lambda$ and $\mu$ parameters as the Petersen graph but not regular.}
    \label{star10Graph}
\end{figure}
    
Although this is a counterexample, its nature is similar to trivial strongly regular graphs.
\begin{lemma}
    Let $G = (V,E)$ be a graph with 10 vertices with parameters $\lambda = 1$ and $\mu = 2$, and a vertex $v$ with degree $deg(v) = 3$. Then $G$ must be 3-regular.
\end{lemma}
\begin{proof}
    Every two non-adjacent vertices $x$ and $y$ share one neighbour. Since each vertex $z \in N(x) - N(y)$ must share exactly $\mu = 1$ neighbour with $y$, there must be a matching among their neighbours.

    There must be 6 edges among the two sets $V - N[v]$ and $N(v)$. For every pair of $v$'s neighbours, as they are non-adjacent, there must be a matching between their neighbours. Ergo, each of them is connected to two vertices in $V - N[v]$, and all of them have a degree of 3. By continuously applying this, we can infer the graph is 3-regular.
\end{proof}

Although the Petersen graph served as a counterexample, the statement holds for $\mu > 1$:
\begin{theorem} \label{removeKregular}
    Let $G = (V, E)$ denote a graph with $n$ vertices. If the graph has $\lambda$ and $\mu$ parameters, and $\mu > 1$, then the graph is regular and, consequently, strongly regular.
\end{theorem}

\begin{proof}
Let $u$ and $v$ be two adjacent vertices. They share $\lambda$ neighbours. Let $e_v$ be the number of edges between $N(v) \cup N(u)$ and $N(v) - N(u)$, and let $e_u$ be the number of edges between $N(v) \cup N(u)$ and $N(u) - N(v)$. Also, $e_\lambda$ is the number of edges within the set $N(v) \cup N(u)$ and $e_\mu$ is the number of edges between $N(u) - N(v)$ and $N(v) - N(u)$. 

\begin{figure}
    \centering
    \includegraphics[width=0.7\textwidth]{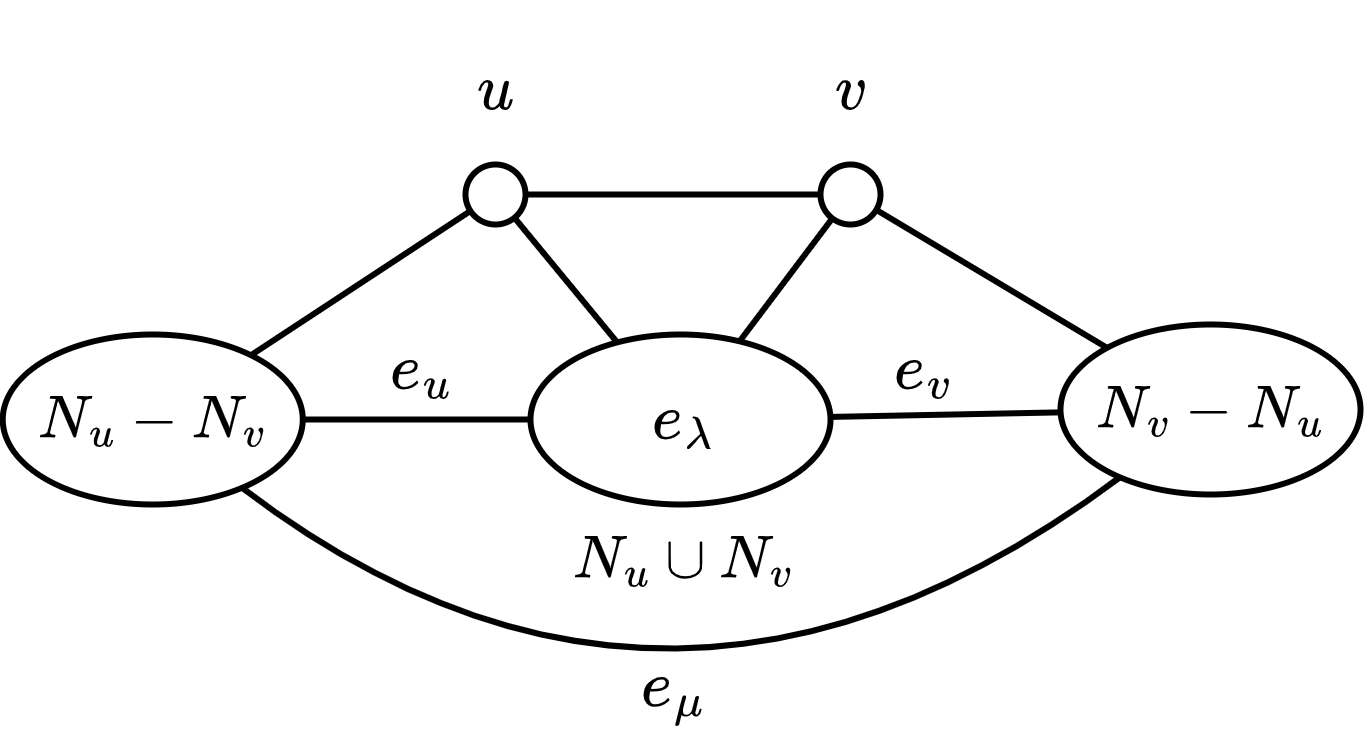}
    \caption{An adjacent pair of vertices $\{u,v\}$ in a graph with $\lambda$ and $\mu$ parameters.}
    \label{theorem_lambda_mu_parameters}
\end{figure}

Due to the $\lambda$ parameter, the induced subgraphs $N(u)$ and $N(v)$ form a $\lambda$-regular graph, and thus,
$$
e_v + 2e_{\lambda} = e_u + 2e_{\lambda} = \lambda(\lambda - 1).
$$ We can infer $e_u = e_v$. Also, because of the $\mu$ condition
\begin{equation*}
(\mu - 1)(deg(v) - \lambda - 1) = e_v + e_\mu \text{, and }(\mu - 1)(deg(u) - \lambda - 1) = e_u + e_\mu,
\end{equation*}
must hold. So,
\begin{equation*}
    (\mu - 1)(deg(u) - deg(v)) = 0;
\end{equation*}
hence $\mu > 1$, the degrees of adjacent vertices $u$ and $v$ are the same. Since the graph is connected, by continuing this, we infer all vertices must have the same degree, and this means the graph is regular.
\end{proof}

One question that arises is whether other parameters, $\lambda$ and $\mu$, are also redundant. That is, if we suppose graph $G$ is regular and every two neighbours share $\lambda$ neighbours, then all non-neighbours share the same number of neighbours $\mu$, and vice-versa.

The triangular view of the Berlekamp–Van Lint–Seidel graph is a counter-example to this. This graph has parameters $(891,30,9, \leq 3)$ and is not strongly regular. In general, for a 30-regular 891 vertex graph with $\lambda = 9$, we know no such strongly regular graph exists as no integer $\mu$ can satisfy the equation in Theorem \ref{basicDoubleCounting}.

\chapter{Studying the experiments} \label{experiments}
In this chapter, we will review the experiments we had, based on how we encoded the problems in Chapter \ref{SATSearch}. The full set of codes and experiments can be found on \href{https://github.com/AliKeramatipour/thesis}{GitHub}. The main SAT solver used was a pseudo-boolean solver named \syntactic{clasp} \cite{CLASP}.

Multiple ~12-hour tests were run on parameter sets of strongly regular graphs with less than 65 vertices. Meanwhile, multiple dedicated SAT solver processes also searched for a $(99,14,1,2)$ strongly regular graph. Considering the granted resources, a longer run was not possible.

To make the final comparison, we used pseudo-boolean encodings instead of CNF encodings, as the latter are too large and impractical. Also, non-linear instances were chosen over linear instances. SAT solvers linearise their non-linear inputs, and since the SAT solvers linearise inputs based on their needs, they will give us better and more efficient results. For example, the strongly regular graph with parameters $(26,10,3,4)$ was found $\approx 2.5$ times faster with a non-linear method.

During the execution of these tasks, in some instances, the SAT solver generated results for larger graphs but led to no results for smaller ones. We theoretically analyse these results.

\section{Automorphism}
The order of automorphism group $|Aut(G)|$ serves us as a measurement of a graph's symmetry. As we studied spectrum and automorphisms in Chapter \ref{preliminaries}, all strongly regular graphs with the same set of parameters are cospectral. However, non-isomorphic instances of these graphs exist, which implies not all cospectral graphs are isomorphic. These non-isomorphic instances lead to multiple automorphism groups for strongly regular graphs with the same parameter set.

Although in the realm of strongly regular graphs, spectrums do not provide us with useful information about isomorphisms and automorphism groups, they are a useful tool to determine non-isomorphic graphs in general.

We now study a few parameter sets for strongly regular graphs that were successfully found using SAT solvers. Our initial analysis focuses on smaller graphs \ref{smallerGraphs}, followed by an examination of bigger graphs \ref{largerGraphs}. At last, we study those that were not attainable using our methods \ref{notFound}.

\subsection{Small graphs} \label{smallerGraphs}
\subsubsection{Schläfli graph}
The Schläfli graph, named after Ludwig Schläfli, is a strongly regular graph with parameters $(27, 16, 10, 8)$. The building method of this graph is geometrical, but we won't delve into it here. Based on its construction method, it can be guessed that this graph has a considerable automorphism group order, specifically 51840. This is notably large when compared to other strongly regular graphs with a close number of vertices. For example, Paley(25) has an automorphism group of order 600.

\subsubsection{Paulus graphs}
There are two sets of non-isomorphic strongly regular graphs known as Paulus graphs. One set has parameters of $(25,12,5,6)$ and consists of 15 graphs, one of which is isomorphic to Paley(25), while the other set has parameters of $(26,10,3,4)$ and contains 10 graphs.

For parameters $(26,10,3,4)$, the largest automorphism group is of order 120.

\subsubsection{Parameters (28,9,0,4)}
Krein bounds, named after Mark Krein, consist of two inequalities that strongly regular graphs must meet in relation to their eigenvalues. They stem from Krein parameters.

\begin{theorem}[Krein bounds]\label{KreinTheorem}
    A strongly regular graph with eigenvalues $r$, $s$, and $k$, where $k$ is its regularity, must satisfy the following bounds:
    \begin{align*}
    (r + 1)(k + r + 2rs) &\leq (k + r)(s + 1)^2,\\
    (s + 1)(k + s + 2rs) &\leq (k + s)(r + 1)^2.
    \end{align*}
\end{theorem}

Parameters $(28,9,0,4)$ cannot form a strongly regular graph. Applying the same approach as we used in \ref{refRemoveParameter} gives us a spectrum of $\{9^1, 1^{21}, (-5)^6\}$. However, using the Krein bounds, it can be deduced that these values are not possible for a strongly regular graph.

\subsubsection{Chang graphs}
In \cite{chang1960association}, Chang showed four non-isomorphic strongly regular graphs with the parameter set $(28,12,6,4)$ exist. One of these graphs, named the Triangular graph $T(8)$, has an automorphism group of order 40320. We study Triangular graphs further in \ref{triangularGraphs}.

In the comparison chart shown in Table \ref{around27}, we have analysed the previous graphs and the time taken by the SAT solver to arrive at a solution. The parameter set $(28,9,0,4)$ failed to yield any result even after approximately 12 hours. This example highlights a scenario where we are aware that the parameters are not feasible, using a close Krein bound, and the SAT solver cannot reach a final state. We can also see that the Schläfli graph parameter set received an answer in less than two seconds while finding a Paulus graph with fewer vertices takes longer. The large automorphism group of the $T(8)$ graph is another example where our SAT solver performed fast.

\begin{table}[htbp]
\centering
\renewcommand{\arraystretch}{1.2}
\begin{tabular}{r|rrrr|crcr}
\hline
 & $n$ & $k$ & $\lambda$ & $\mu$ & time(s) & result\\
\hline
1  & 27 &  16 & 10 & 8 & 1.982 & $\exists$\\
2  & 26 &  10 & 3 & 4 & 10.584 & $\exists$\\
3  & 28 & 9 & 0 & 4 & - & No result\\
4  & 28 & 12 & 6 & 4 & 0.48 & $\exists$\\
\hline
\end{tabular}
\caption{Parameter sets with around 27 vertices. Based on the Krein bounds, the set $(28,9,0,4)$ is unsatisfiable, where SAT solvers cannot yield any results.}
\label{around27}
\end{table}

\subsection{Large graphs} \label{largerGraphs}
To further investigate SAT solver's capabilities, we consider its performance on larger graphs. Table \ref{testResultsLarge} contains larger graphs where our SAT solver was able to find a solution.
\begin{table}[htbp]
\centering
\renewcommand{\arraystretch}{1.2}
\begin{tabular}{r|rrrr|rl}
\hline
 & $n$ & $k$ & $\lambda$ & $\mu$ & time(s) & comment\\
\hline
1 & 36 & 10 & 4 & 2 & 10.271 & Rook\\
2 & 36 & 14 & 7 & 4 & 577.235 & Triangular\\
3 & 45 & 16 & 8 & 4 & 1749.487 & Triangular\\
4 & 49 & 12 & 5 & 2 & 2.994 & Rook\\
5  & 50 & 7 & 0 & 1 & 2.773 & Hoffman-Singleton\\
6  & 56 & 10 & 0 & 2 & 13473.1 & Sims-Gewirtz\\
7  & 64 & 14 & 6 & 2 & 19.04 & Rook\\
\hline
\end{tabular}
\caption{Larger strongly regular graphs found using SAT solver}
\label{testResultsLarge}
\end{table}

\subsubsection{Rook's graph}
Rook graphs are the Cartesian product of two complete graphs $K_m \square K_n$. In these graphs, edges represent the legal moves that can be made by a rook chess piece, that is, moving vertically or horizontally.

When $m = n$ and the order of both complete graphs is the same, the graph becomes strongly regular, with a parameter set of $(n^2, 2n - 2, n - 2, 2)$. Each vertex in this graph is connected to $n - 1$ vertices on two complete subgraphs on its coordinate, which means it is $2n - 2$-regular. 

We name the vertices based on their coordinates, i.e. $(i,j)$ is where the $i$-th horizontal and $j$-th vertical graphs meet. If two vertices share a coordinate, they are adjacent and share $n - 2$ neighbours. For two non-adjacent vertices $(i_1,j_1)$ and $(i_2,j_2)$, they share two neighbours, which are $(i_1,j_2)$ and $(i_2,j_1)$. This implies the graph's strong regularity with parameters $(n^2, 2n - 2, n - 2, 2)$.

As we can see, these graphs are highly symmetrical. The order of $Aut(G)$ for a Rook graph $G$ is $2(n!)^2$. Entries 1,4,and 7 of Table \ref{testResultsLarge} are Rook graphs. We can see that even for a very large order of 64, the SAT solver was able to obtain a result. Again, this is due to the SAT solver's symmetry-breaking techniques.

\subsubsection{Triangular graphs} \label{triangularGraphs}
These highly symmetrical graphs can be built in many different ways. One possible approach is using \textit{Line graphs}. Line graphs are similar to the triangular approach we employed in Subsection \ref{triangular_view}. A line graph of a graph $G$, denoted with $L(G)$, represents the adjacency of \textit{edges}. Two edges are adjacent if they share a vertex.

Triangular graphs can be defined as the Line graph of a complete graph, i.e. $L(K_n)$. Let us now focus on how this graph is strongly regular. This graph consists of $\binom{n}{2}$ vertices, which represent the edges of $K_n$. In $K_n$, each edge shares a vertex with $2n - 2$ edges, and every two edges that share a vertex are adjacent to $n - 2$ other edges. If two edges do not share a vertex, name them $uv$ and $xy$, which means they are non-adjacent, they share four neighbours. These neighbours are edges $xu$, $xv$, $yu$, and $yv$. Therefore, the graph obtained is a strongly regular graph with parameters $(\binom{n}{2}, 2n - 2, n - 2, 4)$. These graphs are also sometimes referred to as \textit{perfect line graphs}.

The way we have built this graph already shows the high order of its automorphism group. The $Aut(L(K_n))$ contains $n!$ permutations, which are based on the initial $K_n$ graph.

There are more ways to build these graphs as well. For example, the complement of $T(n) = L(K_n)$ can be represented by \textbf{Kneser graphs}. A Kneser graph $K(n,k)$ is a graph whose vertices are $k$-element subsets of a set with $n$ elements. Two vertices are connected if their respective subsets are disjoint. The Petersen graph is also the Kneser graph $K(5,2)$, as shown in Figure \ref{petersenKneser}.

\begin{figure}
  \centering
  \includegraphics[width=0.6\textwidth]{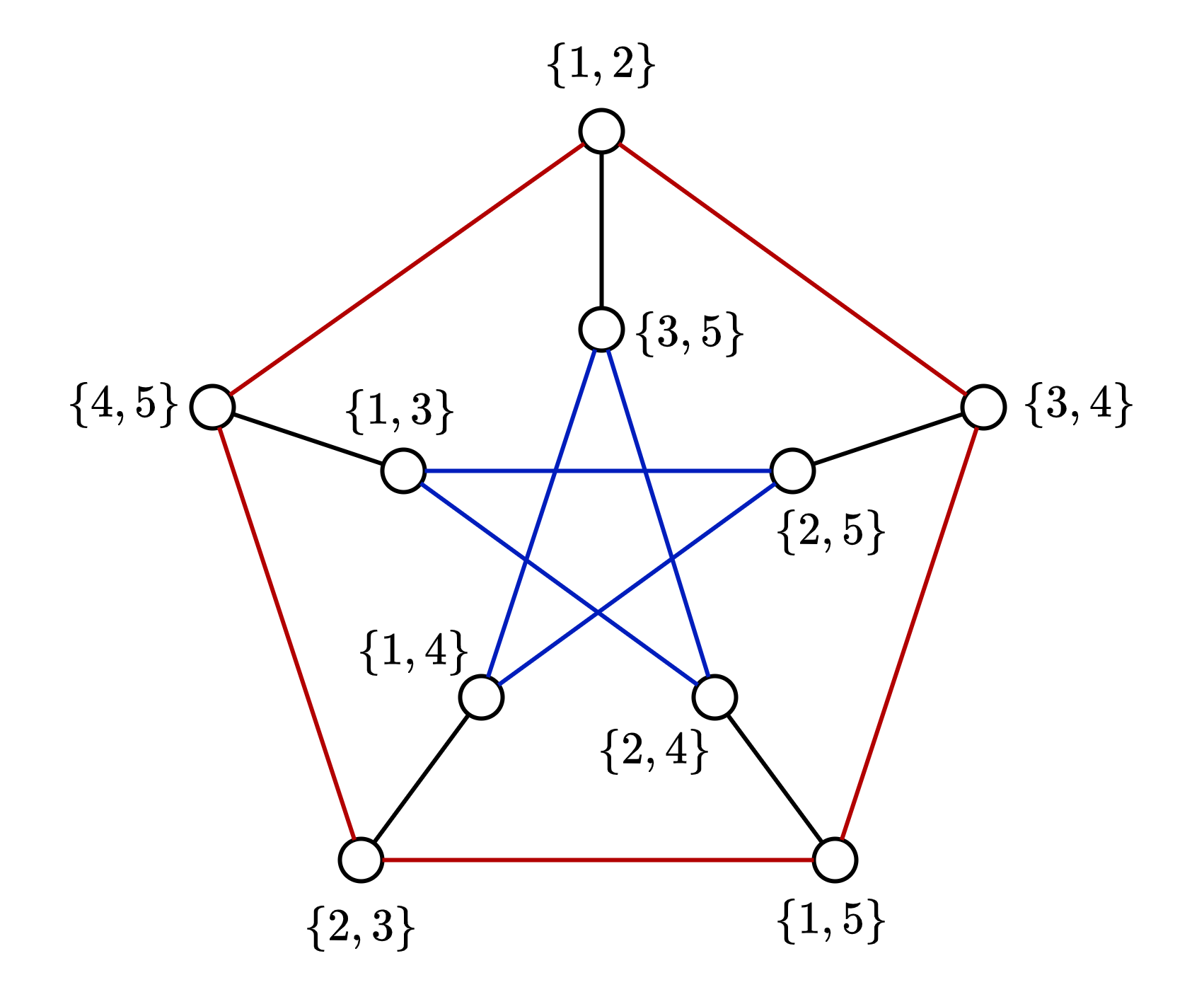}
  \caption{The Petersen or $K(5,2)$ graph. Two vertices are adjacent if their respective sets share no element.}
  \label{petersenKneser}
\end{figure}

The adjacency matrices of triangular graphs $T(n)$ are uniquely determined, up to isomorphism, by their spectrum and eigenvalues when $n$ is not 8. The other three Chang graphs have the same parameters as $T(8)$.

\subsubsection{Hoffman-Singleton}
Only one strongly regular graph with parameters $(50,7,0,1)$ exists. This graph is known as the Hoffman–Singleton graph. Its automorphism group is of order 252000.

\subsubsection{Sims-Gewirtz graph}
This is a unique, strongly regular graph with parameters $(56,10,0,2)$. Its automorphism group is of order 80640.

The results show that all the graphs that were found have large automorphism groups, which makes them highly symmetrical. SAT solvers, by applying their symmetry-breaking techniques, are able to reduce the search space.

\subsection{Graphs that were not found} \label{notFound}
We previously saw that for the set of parameters $(28,9,0,4)$, the SAT solver failed to yield any results while it was already unsatisfiable. We now study satisfiable parameters where the SAT solver still fails. Do they all have small automorphism groups?

\begin{table}[htbp]
\centering
\renewcommand{\arraystretch}{1.2}
\begin{tabular}{r|rrrr|rl}
\hline
 & $n$ & $k$ & $\lambda$ & $\mu$ & comment\\
\hline
1 & 29 & 14 & 6 & 7 & Paley(29)\\
2 & 35 & 16 & 6 & 8 & -\\
3 & 36 & 14 & 4 & 6 & -\\
4 & 36 & 16 & 6 & 6 & -\\
5 & 36 & 25 & 16 & 20 & $\overline{K_6 \square K_6}$\\
5 & 37 & 18 & 8 & 9 & Paley(37)\\
6 & 41 & 20 & 9 & 10 & Paley(41)\\
\hline
\end{tabular}
\caption{Parameter sets where SAT solver was unable to find a solution.}
\label{failedTests}
\end{table}

\subsubsection{Paley graphs}

We studied the construction method of Paley graphs in \ref{paleyGraphs}. Consider a Galois field $GF(q)$, where $q = p$, a prime number, and the set of its quadratic residues $QR(q)$. Define a function $f_{a,b}: GF(q) \rightarrow GF(q)$:
\begin{equation*}
    f_{a,b}(x) = ax + b,
\end{equation*}
where $a \in QR(q)$ and $b \in GF(q)$. There are $\frac{1}{2}(q - 1)$ quadratic residues; hence, coefficients $a$ and $b$ have $\frac{1}{2}(q - 1)q$ possible assignments to form linear functions. In a Paley graph $G = (GF(q), E)$, two vertices are connected if their difference is a quadratic residue. Since the product of two quadratic residues is itself a quadratic residue,
\begin{equation*}
    (x - y) \in QR(q) \Leftrightarrow (ax + b - ay - b) \in QR(q)
\end{equation*}
holds, as $a,(x-y) \in QR(q)$. This gives us $\frac{1}{2}(q - 1)q$ automorphisms. According to Carlitz's theorem \cite{carlitz1960theorem}, the Galois Field $GF(q)$ of a prime power $q = p^d$ possesses $d$ automorphisms.  Thus, by combining these automorphisms, an automorphism group of order $\frac{d}{2}(q - 1)q$ is found for Paley(q). It can also be proved no more automorphisms exist.

In Table \ref{failedTests}, we can see three sets of parameters that Paley graphs can satisfy. For parameter set $(29,14,6,7)$, a total of 41 non-isomorphic graphs have been found by computer search. A list of these graphs is compiled \href{http://www.maths.gla.ac.uk/~es/srgraphs.php}{here} by Ted Spence. Therefore, the Paley(41) graph was not the only answer to this SAT problem, and even more answers could have been found.

Based on these observations, we can see even for highly structured graphs like the Paley set, SAT solver fails to solve the problem. The two sets of graphs where SAT solver did not fail both have exponential automorphism orders, while the Paley graphs' automorphism set is polynomial.

\subsubsection{Rook graph's complement}
Rook graphs possessed a large automorphism group. However, SAT solver was unable to solve this instance $\overline{K_6 \square K_6}$, which is the Rook graph's complement, with high exponential symmetry.

\subsubsection{Graphs with 36 vertices}
McKay and Spence studied two parameter sets of strongly regular graphs in \cite{mckay2001classification}, which are $(36,15,6,6)$ and $(36,14,4,6)$.

For $(36,14,4,6)$, a total of 180 non-isomorphic instances were found. The parameter set $(36,15,6,6)$ has 32548 possible graphs. There are several possible solution instances available, however, they lack high symmetry, which poses a challenge for the SAT solver. The solver is unable to overcome this, and its performance is similar to an exhaustive search.

\subsubsection{The Conway-99 problem}
Based on our observation, we can infer that if $(99,14,1,2)$ strongly regular graphs exist, then each of these graphs will likely have a very small automorphism group.

\section{Alternative and additional methods}

\subsection{Additional constraints}
It is beneficial to have additional constraints to guide the SAT solver towards the correct solution. During our experiments to study parameter redundancy \ref{parameterRedundancy}, we ran multiple tests to see whether the SAT solver gives us a strongly regular graph based on only two parameters. For parameters $(28,12,6,4)$, based on $\lambda = 6, \mu = 4$ and $k= 12,\mu = 4$, and $k = 12,\lambda = 6$ the time it took was approximately $\approx 2.5$ minutes, $\approx 2.5$ hours, $\approx 9$ hours respectively.

After analysing this, it may be advantageous to include more restrictions in our SAT formula. A potential approach is to introduce constraints related to vertex colouring. It is common knowledge that all graphs can be coloured using $\Delta(G) + 1$ colours. Nevertheless, it may not be ideal to impose a strict limitation as the SAT solver will need to tackle both the chromatic number $\chi(G)$ problem and find the strongly regular graphs. In some instances, leaving the SAT solver to colour the graph with $2\Delta(G)$ colours sped up the process.

Also, more feasibility conditions exist for the class of strongly regular graphs as we studied them. Studying these conditions in more detail and encoding them into the SAT formula might make SAT solvers yield results faster.

\subsection{MAX-SAT and almost strongly regular graphs}
SAT solvers used in our studies are strict on satisfying all clauses. In order to obtain an almost strongly regular graph, we can remove some of the neighbouring or regularity constraints. We observed that even for small unsatisfiable constraints, they did not reach any terminal state. This could be the result for the Conway-99 parameters as well.

The maximum boolean satisfiability problem, \textbf{MAX-SAT}, asks for the maximum number of clauses that can be satisfied in the SAT formula. This problem is categorised as an NP-Hard problem. Although exact MAX-SAT solvers exist, approximate MAX-SAT solvers are of our interest here. While the solver is searching for the solution, we want the solver to provide us with the best solution it has found so far.

The combination of SAT and MAX-SAT problems gives us the \textbf{partial MAX-SAT problem}. In these formulas, some clauses are categorised as \textit{hard} clauses and others are categorised as \textit{relaxable} or \textit{soft} ones. The solver is forced to satisfy all hard clauses, while it looks for the maximum number of soft clauses it can satisfy.

By employing the MAX-SAT method, we can find solutions that are almost strongly regular. For example, we force the SAT solver to return a $k$-regular graph while it tries to satisfy the maximum number of neighbouring conditions. SAT solvers that can provide us with the best solution they have found so far can be great to use. Examining where they hit a dead-end and need to backtrack might give us new insights into this problem and let us come up with other general feasibility conditions.

The following theorem shows us another approach by removing the $\lambda$ conditions.
\begin{theorem}
    In a partial MAX-SAT formula, for a graph $G$ containing 99 vertices, if the 14-regularity and the $\mu = 2$ conditions are defined as hard clauses, and a $(99,14,1,2)$ strongly regular graph exists, then minimising the number of Claw subgraphs (as soft clauses) must satisfy the $\lambda = 1$ condition. If it does not, no such strongly regular graph exists.
\end{theorem}
\begin{proof}
    In the graph $G$ obtained as a solution, consider an arbitrary vertex $v$. In our first step, we demonstrate $N(v)$ has a total of 7 edges. There are 84 vertices in $V - N[v]$, each connecting to exactly two vertices in $N(v)$ ($k = 14$ and $\mu = 2$ hard clauses). This implies $84 \times 2$ outgoing edges from $N[v]$ exist. Also, there are 14 edges among $v$ and $N(v)$. We can infer that
    \begin{equation*}
        \frac{14 \times 14 - 84 \times 2 - 14}{2} = 7
    \end{equation*}
    edges must be in $N(v)$.

    Each non-neighbour pair $\{u, w\} \subset N(v)$ must be connected to exactly one vertex in $V - N[v]$, outside $v$'s neighbourhood. Otherwise, there must exist a non-neighbouring pair of vertices sharing three neighbours, which breaks the $\mu = 2$ condition. Based on this, we can infer no Cherry induced subgraphs ($K_{1,2}$) can exist within $N(v)$. Therefore, in subgraph $N(v)$, all components must be complete ($K_i$ for $i \in \{1,2,3,4\}$). Figure \ref{N14neighbourhood} contains all possible $N(v)$ induced subgraphs.

    We now count the number of Claw subgraphs in possible subgraphs. Let $c_i$ be the number of $K_i$s in $N(v)$. The number of Claws is obtained by
    \begin{equation*}
        \binom{14}{3} - 12 \times c_2 - 11 \times 3c_3 - 10 \times 6c_4 - c_3 - 4c_4,
    \end{equation*}
    which would be minimised in case $a)$ of Figure \ref{N14neighbourhood}. In this case, graph $G$ has the $\lambda = 1$ parameter and is strongly regular.

\end{proof}

\begin{figure}
    \centering
    \includegraphics[width=\textwidth]{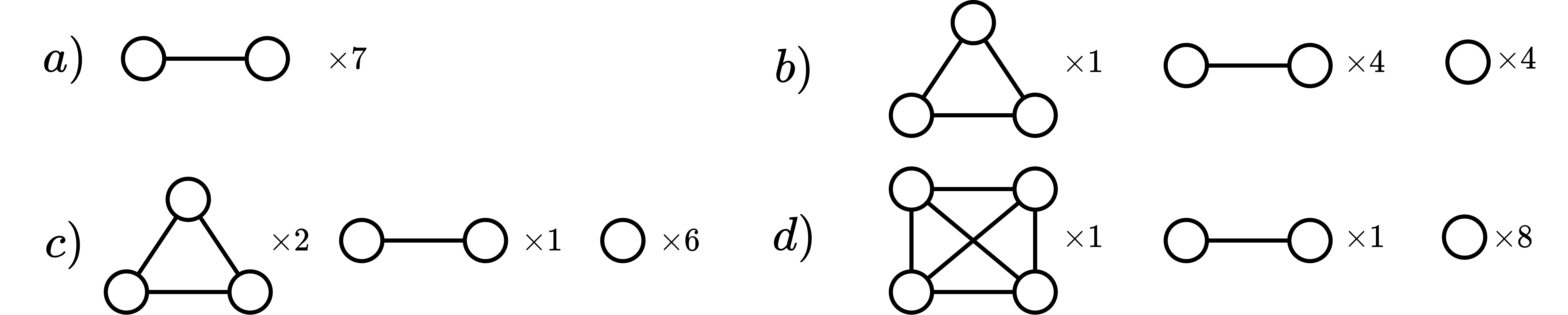}
    \caption{Possible $N(v)$ induced subgraphs}
    \label{N14neighbourhood}
\end{figure}

\chapter{Conclusions}
In ``Approaching the Conway-99 problem using SAT solvers", we studied strongly regular graphs both combinatorically and algebraically. We studied structures similar to the graph that the Conway-99 problem asks for and proved these patterns cannot hold in any solution for this problem.

To search for this problem, we decided to use SAT solvers to see whether they could find any solution to this problem. Therefore, we encoded this problem and transformed it into a SAT formula. As the SAT solver failed, we studied the reason behind this failure theoretically.

Countless ideas came to our minds on this beautiful research path. Unfortunately, we hit the wall of a deadline; otherwise, there is much to do research on. Some of these ideas will be listed here. The topics and ideas we gather from these sources will be valuable for our future research endeavours.

\section{Future work}
\begin{itemize}
    \item As we studied, the structure of strongly regular graphs could be applied to many areas. Can we use this structure to develop better SAT solvers? Are these structures useful for detecting symmetries faster? Are their irregular structures useful for solving asymmetries faster?
    \item While approaching this, because of the short research period, we considered SAT solvers as black boxes. In the future, by studying the methods used in SAT solvers and combining them with our problem-specific methods, we can generate a better search tool for this problem.
    \item In \ref{noPaley9Subgraph}, we conjectured that a Paley(9) subgraph cannot exist in a possible $(99,14,1,2)$ strongly regular graph. If a Paley(9) subgraph appears, it would be a dominating set of the whole graph. To approach this problem, we can consider dominating sets and see whether we can have a 9-vertex Paley(9) dominating set.
    \item In Table \ref{lambda1mu2Table}, we listed the possible instances of strongly regular graphs with properties $\lambda = 1$ and $\mu = 2$. The Berlekamp–Van Lint–Seidel graph was built based on $GF(3)$ and Paley(9). Another possible graph to be discovered is a $(6273,112,1,2)$ strongly regular graph. 
    
    Considering the prime factorisation of 6273
    \begin{equation*}
        6273 = 3^2 \times 17 \times 41,
    \end{equation*}
    we can see all three factors are congruent to 1 modulo 4, which means three Paley(9), Paley(17), and Paley(41) graphs exist. Can we use these three graphs to build a $(6273,112,1,2)$ strongly regular graph?
    \item In \ref{p9patternImpossible} we proved a $(99,14,1,2)$ strongly regular graph cannot follow the Paley(9) pattern. We used this in our SAT formula and preset the value of multiple edge boolean variables. It is highly likely we can prove a bound on the number of vertices following this pattern. Using this, we can reduce our search space by presetting more edge boolean variables.
    \item In 1990, Paul Seymour conjectured that in every directed graph, there is a vertex whose second neighbourhood, that is, the set of vertices of minimum distance 2, is as large as its first neighbourhood. Fisher, in \cite{fisher1996squaring}, partially proved this by proving it for only tournaments, a simple, directed, complete graph. The graph $K_n$ is a highly structured graph and is a trivial strongly regular graph.

    A question that arises is whether we can generalise this for all strongly regular graphs. This seems likely as they are highly structured. We proved this theorem for strongly regular graphs with $\lambda = 1$.
\end{itemize}

\section{Acknowledgements}
I am grateful for the guidance and support provided by Prof. Anuj Dawar throughout this project. I would like to express my heartfelt appreciation to my family, friends, and the esteemed members of Churchill College who motivated and assisted me on this journey.

\label{lastcontentpage} % end page count here
\label{lastpage}

\bibliographystyle{plain}
\bibliography{references}

\appendix
\chapter{Runtime test experiments}
{\small\begin{center}
		\begin{longtable}[htb]{rrrr|rl}
\hline
$n$ & $k$ & $\lambda$ & $\mu$ & time(s) & Comment\\
\hline
5& 2& 0& 1& 0.0 \\
9& 4& 1& 2& 0.001 \\
10& 3& 0& 1& 0.001 \\
10& 6& 3& 4& 0.002 \\
13& 6& 2& 3& 0.004 \\
15& 6& 1& 3& 0.007 \\ 
15& 8& 4& 4& 0.007 \\ 
16& 5& 0& 2& 0.006 \\
16& 6& 2& 2& 0.012 \\
16& 9& 4& 6& 0.012 \\
16& 10& 6& 6& 0.049 \\
17& 8& 3& 4& 0.407 \\
21& 10& 3& 6& 0.068 \\
21& 10& 4& 5& - & $\emptyset$ \\
21& 10& 5& 4& 0.32 \\
25& 8& 3& 2& 1.717 \\
25& 12& 5& 6& 67.719 \\
25& 16& 9& 12& 18.991 \\
26& 10& 3& 4& 10.584 \\
26& 15& 8& 9& 22.224 \\
27& 10& 1& 5& 0.587 \\
27& 16& 10& 8& 1.982 \\
28& 9& 0& 4& - & $\emptyset$ \\
28& 12& 6& 4& 0.48 \\
28& 15& 6& 10& 0.233 \\
28& 18& 12& 10& - & $\emptyset$\\
29& 14& 6& 7& - \\
33& 16& 7& 8& - & $\emptyset$ \\
35& 16& 6& 8& - \\
35& 18& 9& 9& - \\
36& 10& 4& 2& 10.271 \\
36& 14& 4& 6& - \\
36& 14& 7& 4& 577.235 \\
36& 15& 6& 6& - \\
36& 20& 10& 12& - \\
36& 21& 10& 15& - \\
36& 21& 12& 12& - \\
36& 25& 16& 20& - \\
37& 18& 8& 9& - \\
40& 12& 2& 4& - \\
40& 27& 18& 18& -\\
41& 20& 9& 10& - \\
45& 12& 3& 3& - \\
45& 16& 8& 4& 1749.487 \\
45& 22& 10& 11& - \\
45& 28& 15& 21& - \\
45& 32& 22& 24& - \\
49& 12& 5& 2& 2.994 \\
49& 16& 3& 6& - & $\emptyset$ \\
49& 18& 7& 6& - \\
49& 24& 11& 12& - \\
49& 30& 17& 20& - \\
49& 32& 21& 20& - & $\emptyset$ \\
49& 36& 25& 30& - \\
50& 7& 0& 1& 2.773 \\
50& 21& 4& 12& - & $\emptyset$ \\
50& 21& 8& 9& - \\
50& 28& 15& 16& - \\
50& 28& 18& 12& - & $\emptyset$ \\
50& 42& 35& 36& - \\
53& 26& 12& 13& - \\
55& 18& 9& 4& - \\
55& 36& 21& 28& - \\
56& 10& 0& 2& 13473.1 \\
56& 22& 3& 12& - & $\emptyset$ \\
56& 33& 22& 15& - & $\emptyset$ \\
56& 45& 36& 36& - \\
57& 14& 1& 4& - & $\emptyset$ \\
57& 24& 11& 9& - \\
57& 28& 13& 14& - & $\emptyset$ \\
57& 32& 16& 20& - \\
57& 42& 31& 30& - \\
61& 30& 14& 15& - \\
63& 22& 1& 11& - & $\emptyset$ \\
63& 30& 13& 15& - \\
63& 32& 16& 16& - \\
63& 40& 28& 20& - & $\emptyset$ \\
64& 14& 6& 2& 19.04 \\
64& 18& 2& 6& - \\
64& 21& 0& 10& - & $\emptyset$ \\
64& 21& 8& 6& - \\
64& 27& 10& 12& - \\
64& 42& 26& 30& - \\
64& 42& 30& 22& - & $\emptyset$ \\
64& 45& 32& 30& - \\
64& 49& 36& 42& - \\
\hline
\caption{All SAT tasks. SAT solver \red{clasp} was used. The character `-' means after $\approx 12$ hours, no result was obtained. Character $\emptyset$ denotes that no graphs with the respective parameters exist.}
	\label{allExperiments}
\end{longtable}
  \end{center}
}

\end{document}